\documentclass[12pt]{article}

\usepackage[font=small,format=plain,labelfont=bf,textfont=it]{caption}
\usepackage{setspace}
\usepackage{apptools}

\usepackage{fullpage, verbatim, amsfonts, amsmath, amssymb, amsthm, latexsym, setspace, graphicx}
\usepackage[T1]{fontenc}
\usepackage[latin9]{inputenc}
\usepackage{geometry}
\usepackage{xcolor}
\usepackage{natbib}
\usepackage{overpic}
\usepackage{tikz}
\usepackage{soul}
\usepackage{caption}
\usepackage{graphicx}
\usepackage{subcaption}
\usepackage[normalem]{ulem}

\usepackage[french, english]{babel}
\usepackage{hyperref,bookmark}

\hypersetup{linkcolor=MyDarkBlue, citecolor=MyDarkBlue, colorlinks=true} 
	\definecolor{MyDarkBlue}{rgb}{0,0.08,0.45}
	\definecolor{MyDarkGreen}{rgb}{0,0.55,0.08}

\theoremstyle{plain}
\newtheorem{thm}{Theorem}
\newtheorem{lem}{Lemma}
\newtheorem{cor}{Corollary}
\newtheorem{prop}{Proposition}

\theoremstyle{definition}
\newtheorem{rem}{Remark}

\AtAppendix{\counterwithin{lem}{section}}
\AtAppendix{\counterwithin{prop}{section}}
\AtAppendix{\counterwithin{thm}{section}}
\AtAppendix{\counterwithin{ex}{section}}
\AtAppendix{\counterwithin{cor}{section}}

\newcommand{\bl}{\textcolor{blue}}
\newcommand{\red}{\textcolor{red}}

\newcommand{\argmax}{\mathop{\mathrm{argmax}}}

\newcommand\cites[1]{\citeauthor{#1}'s \citeyearpar{#1}}

\newcommand{\finex}{\leavevmode\unskip\penalty9999 \hbox{}\nobreak\hfill\quad\hbox{$\blacktriangle$}}

\begin{document}

	\title{Auctions as Experiments\thanks{Frick: Princeton University (mfrick@princeton.edu); Iijima: Princeton University (riijima@princeton.edu); Ishii: Pennsylvania State University (yxi5014@psu.edu); Wu: Microsoft Research (nicktwu@stanford.edu). For helpful comments, we thank Yu Awaya, V.\ Bhaskar, Dirk Bergemann, Ben Brooks, Sylvain Chassang, Phil Haile, Emir Kamenica, Jakub Kastl, Ilan Kremer, Vijay Krishna, Alessandro Lizzeri, Stephan Lauermann, Stephen Morris, Alessandro Pavan, Wolfgang Pesendorfer, Larry Samuelson, Ran Shorrer, Andy Skrzypacz, Omer Tamuz, Caroline Thomas, Leeat Yariv, as well as numerous seminar and conference audiences. Iijima gratefully acknowledges the financial support from a Sloan Research Fellowship.} }
	\author{Mira Frick \and Ryota Iijima \and Yuhta Ishii \and Nicholas Wu}
	
	\date{This version: \today \\ \vspace{2mm} First posted version: February 9, 2026}
	\maketitle

	\begin{abstract}

Consider an auction with buyers whose values depend on an underlying state (e.g., market fundamentals). How does the auction format shape the information that buyers' bids reveal about the state? We recast auctions as statistical experiments and compare different auction formats in terms of the (Lehmann) informativeness of the induced experiments. Our main finding is that among a large class of auctions (e.g., $k$th-price, all-pay), the first-price auction is the most informative. As a result, this auction guarantees the highest payoffs to a decision-maker who uses the information revealed by buyers' bids in a monotone decision problem (e.g., a prediction problem or the choice of a reserve price in a future auction).

\medskip

\noindent {\it Keywords:} auctions, statistical experiments, Lehmann order.

\end{abstract}

\onehalfspacing

\newpage
	
\section{Introduction}

Consider a group of buyers who compete for a good in an auction. In many settings, buyers' values for the good depend on an underlying state of the world about which they are better informed than an outside observer. Such an observer may seek to infer buyers' information about the state from their auction bids, and use this information to guide a subsequent decision. Think, for example, about Treasury auctions for government debt securities, where buyers' (e.g., financial institutions')
values depend partly on their information about market fundamentals. Inferring this information from bids can be valuable to the auctioneer (i.e., the government) in guiding issuance and debt-management decisions, or to financial investors and forecasters in informing their investment decisions and predictions.\footnote{See, e.g., \cite{treasury2003} for the Treasury's use of past auction bids in guiding issuance decisions, and \cite{boyarchenko2021} for intermediaries' use of clients' bids in informing their own future trades.} Likewise, in internet and procurement auctions, auctioneers may benefit from using the information gleaned from buyers' bids about the value distribution to adjust reserve prices in future auctions.\footnote{See, e.g., \cite{ostrovsky2023} in the context of sponsored-search auctions and \cite{whited2025} in the context of auctions for highway improvement projects.}

In this paper, we ask how the auction format shapes the information that buyers' bids reveal about the state. To this end, we recast each auction as a statistical experiment that, conditional on each state of the world, produces signals---buyers' equilibrium bids---about the state. We compare different auction formats in terms of the informativeness of the induced experiments. Our main finding is that, among a large class of auctions, the first-price auction is the most informative. As a result, this auction guarantees the highest payoffs to a decision-maker (henceforth, DM) who uses auction bids to learn about the state prior to making a decision.

In our model (Section~\ref{sec:model}), $n$ buyers compete for one indivisible good in an auction. Buyers' values $v_i$ for the good are drawn independently from a distribution that depends monotonically on a state of the world $\theta \in \mathbb{R}$. For example, values may take the form $v_i =  \theta + \varepsilon_i$, where $\theta$ represents a common-value component (e.g., market fundamentals or the quality of the good) and $\varepsilon_i$ an idiosyncratic taste shock that is independent across buyers and $\theta$. Buyers are better informed about the state $\theta$ than a DM (e.g., the auctioneer or a third party), who draws inferences about $\theta$ from observing buyers' auction bids. To capture this, our baseline model assumes that buyers perfectly observe $\theta$ (in addition to $v_i$) while the DM does not (Section~\ref{sec:imperfect} relaxes this assumption). To compare different auction formats in terms of the information that buyers' bids reveal about $\theta$, we consider the symmetric monotone equilibria of a large class of \emph{standard auctions}: Buyers are ranked according to their (sealed) bids; the highest bidder wins the good; and all buyers make payments that can depend in fairly general ways on their rank and the entire profile of bids. Simple examples include $k$th-price auctions, where the highest bidder pays the $k$th-highest bid and others pay nothing; and all-pay auctions, where all bidders pay their own bids.

The equilibrium bid distributions of each such auction constitute a statistical experiment: Conditional on each state $\theta$, the DM observes $n$ independent signals about $\theta$, viz., buyers' equilibrium bids. (As Remark~\ref{rem:partial-observation} discusses, the analysis goes through if the DM observes only a subsample or selected statistics of bids.) We compare the informativeness of these experiments in terms of \cites{lehmann1988} accuracy order. This formalizes a sense in which one auction is \emph{more accurate} than another if the equilibrium bid distributions under the first auction are more sensitive to increases in $\theta$ than those under the second auction. Lehmann-accuracy is a widely used relaxation of \cites{blackwell} informativeness order in settings with one-dimensional signals and states. Although the accuracy order is less conservative than the Blackwell order, it still admits a foundation in terms of robust payoff improvements in a large class of decision problems (see below).

Section~\ref{sec:analysis} analyzes the accuracy order over auctions. We begin by noting that, under any auction, increases in the state $\theta$ affect equilibrium bid distributions through two channels: First, by assumption, the value distribution increases with $\theta$. Second, even fixing a buyer's value $v_i$, her bid may adjust to this increase in the value distribution, as this means that she faces a more competitive pool of other bidders. Crucially, the magnitude and sign of this second, \emph{strategic responsiveness effect} depend on the auction format, and Lemma~\ref{lem:competition} shows that an auction is more accurate precisely if it induces a more positive strategic responsiveness effect.

This observation implies that the first-price auction is more accurate than several canonical auctions (Proposition~\ref{prop:example}). To see the idea, first consider the second-price auction. Here, each buyer bids her own value regardless of the state, so observing these bids reveals buyers' values to the DM. In contrast, under the first-price auction, a DM who does not know $\theta$ cannot in general perfectly back out buyers' values from their bids. Nevertheless, first-price bids are even more informative about the underlying value distribution/state than directly observing buyers' values. The reason is that the first-price auction induces a positive strategic responsiveness effect: each buyer type increases her bid in response to an increase in opponents' values. This renders first-price bids more sensitive to increases in $\theta$ than buyers' values themselves, which makes the first-price auction more accurate than the second-price auction.  We also show that the first-price auction is more accurate than the third-price auction, where strategic responsiveness is negative (i.e., buyers reduce their bids in response to increases in opponents' values), as well as the all-pay auction, where strategic responsiveness is dampened by the fact that even non-winning bidders have to make non-zero payments.

Our main result (Theorem~\ref{thm:main}) is that the first-price auction is in fact more accurate than \emph{any} well-behaved standard auction: Unlike the canonical auctions above, arbitrary standard auctions need not admit closed-form equilibria; to establish a general ranking despite this difficulty, we impose some additional structure by restricting to auctions whose equilibrium bid distributions are monotone in $\theta$.\footnote{Monotonicity is in the sense of the monotone likelihood-ratio property. This assumption is not needed for the comparison against many canonical auction formats (e.g., second-, third-price, and all-pay), but Remark~\ref{rem:MLRP} exhibits natural settings where it is satisfied.} The proof of Theorem~\ref{thm:main} exploits the revenue equivalence of all standard auctions to relate their bidding functions to first-price bids. The first-price auction has two defining features:  (i) only the winner pays, and (ii) holding fixed a bidder's rank, her payment depends only on her own bid. We show that these two features render the first-price auction most accurate, because they generate a more positive strategic responsiveness effect than under any other auction.

We also derive accuracy rankings among two subclasses of auctions; each preserves one of the two features of the first-price auction, while isolating the informativeness effect of violating the other feature. Among $k$th-price auctions, which preserve feature (i), greater accuracy corresponds to lower $k$ (Proposition~\ref{prop:k}). Among own-pay auctions---where each bidder's payment is some multiple of her own bid (preserving feature (ii) and generalizing first-price and all-pay auctions)---greater accuracy corresponds to higher-ranked bidders paying relatively larger multiples of their bids than lower-ranked bidders (Proposition~\ref{prop:own bid}).

 Section~\ref{sec:payoff} illustrates the payoff implications of our accuracy comparisons of auctions. Suppose that after observing buyers' bids in an auction, the DM faces a decision problem where he chooses an optimal action based on the information about $\theta$ gleaned from these bids. Applying \cites{lehmann1988} classic result to our setting immediately implies that more accurate auctions lead to higher expected payoffs in all \emph{monotone decision problems}, where higher states correspond to higher ex-post optimal actions (Corollary~\ref{cor:lehmann}). We present two simple examples of such decision problems.

First, we consider statistical prediction problems: the DM's payoffs are described by a loss function that is increasing in the distance between his action and the state. How close the DM's expected loss comes to zero provides a quantitative measure of an auction's informativeness. We show that, depending on the value distributions, the informational advantage of first-price over second-price auctions can be arbitrarily large: Indeed, the expected loss under first-price bid observations can be arbitrarily close to zero while that under second-price bid observations is arbitrarily close to the no-information benchmark, and this can happen even when there are arbitrarily many more bidders in the second-price than the first-price auction.
Second, we suppose the DM is the auctioneer himself, who uses today's bid observations to set a reserve price in a future auction with the same value distribution as today. By extending Theorem~\ref{thm:main} to accommodate reserve prices, we show that to maximize total revenue across the two periods, the auctioneer optimally uses a first-price auction in period 1: Under any period-1 reserve price, the first-price auction yields the same period-1 revenue as any other auction but provides more accurate information about $\theta$; this improves period-2 revenue, as the choice of a period-2 reserve price is a monotone decision problem.

Finally, Section~\ref{sec:extensions} extends Theorem~\ref{thm:main} to environments where buyers have interdependent values (Section~\ref{sec:interdependent}) or only observe noisy signals about the state $\theta$ (Section~\ref{sec:imperfect}). Further afield, we also show that the comparison between first- and second-price auctions generalizes to correlated-value settings and multi-unit auctions (Section~\ref{sec:other-extensions}).

\subsection{Related Literature}\label{sec:literature}

We relate to several literatures that study inference in auctions from different angles.

The classic literature on information aggregation in common-value auctions investigates under which conditions the winner's payment (i.e., equilibrium price) reveals the common value as the number of buyers $n \to \infty$ \citep[e.g.,][]{wilson1977, milgrom1979, pesendorfer1997}. Away from this limit, \cite{di2021, di2026} provide conditions under which the winner's payment becomes Lehmann-more accurate as $n$ increases.
Some papers compare specific auction formats in terms of the informativeness of the winner's payment. \cite{hong2004} show that the winner's payments under the second-price and English auctions are asymptotically equally informative as the number of buyers $n \to \infty$, although the latter is more informative for any fixed $n$ \citep[][]{kremer2002}. \cite{kremer2005} show that, for any $k$, the winner's payment under the $(k+1)$th-price auction is asymptotically more informative than under the $k$th-price auction as $n\to\infty$. \cite{awaya2026} compare the informativeness of the winner's payment in first- and second-price auctions (and their multi-unit extensions) for a fixed number of buyers $n$. They provide a condition, in terms of the value distribution and $n$ (and, in the multi-unit case, the number of units $m$), under which the winner's payment in the first-price auction is Lehmann-more accurate. 

In the above papers, buyers' types are one-dimensional signals about the common value, so an outside observer can invert any observed equilibrium bid to recover the corresponding type. As a result, informational differences across auction formats arise because the winner's payments reveal different order statistics of the type distribution---for example, the highest vs.\ the second-highest type under first- vs.\ second-price auctions. Thus, the above comparisons reflect the \emph{statistical} channel that different order statistics of a sample contain different information about the underlying distribution. In contrast, in our paper, buyers observe both their idiosyncratic value and the common state $\theta$ (or a noisy signal thereof), which means that an outside observer who does not know $\theta$ cannot directly back out buyers' values from their equilibrium bids (see Remark~\ref{rem:type-only}). This creates a new, \emph{strategic} channel for informational differences across auctions---the aforementioned strategic responsiveness effect: the auction format affects how a buyer of a given value adjusts her equilibrium bid in response to changes in the state (i.e., the distribution of opponents' values).
To isolate this channel, we assume that the DM observes all $n$ buyers' bids (or some subsample of bids that does not depend on the auction format---see Remark~\ref{rem:partial-observation}) rather than only the winner's payment. We highlight that, for any value distributions and any number of buyers $n$, this channel makes the first-price auction most accurate among all well-behaved standard auctions, as the first-price auction induces the strongest strategic responsiveness effect.

A recent literature studies how an auctioneer can use observations from past auctions to optimize revenue in future auctions \citep[for a survey, see, e.g.,][]{nedelec2022}. \cite{cole2014} and \cite{morgenstern2015} study an auctioneer who uses period-1 observations of buyers' values (e.g., their bids in a second-price auction) to learn about the value distribution and to optimally set a reserve price for a period-2 auction; they analyze the convergence rate of the auctioneer's period-2 revenue to the known-distribution case as the number of period-1 observations grows large.\footnote{\cite{segal2003} and \cite{baliga2003} study revenue-maximizing mechanisms where allocations depend on the information about the value distribution revealed by buyers' reported values.} In Section~\ref{sec:reserve}, we consider a closely related reserve-price choice application, but our interest is in comparing the informativeness of different period-1 auction formats: We highlight that observing first-price bids outperforms directly observing buyers' values, as the strategic responsiveness effect under first-price auctions makes buyers' bids more informative about the value distribution than their values themselves. Some papers study dynamic reserve-price choice (focusing mostly on second-price auctions), where each period's reserve price affects current bid observations and hence the auctioneer's continuation revenue \citep[e.g.,][]{cesa2014}. This gives rise to a tradeoff between exploration and exploitation, a central theme in the literature on experimentation in monopoly pricing \citep[e.g.,][]{rothschild1974}. In Section~\ref{sec:reserve}, the period-1 choice of reserve price also involves this exploration-exploitation tradeoff, but our main point is that the choice of auction format does not: All standard auctions are statically revenue-equivalent, but the first-price auction is most accurate and hence yields the highest continuation revenue.

While we focus on the information that auctions provide about the underlying value distribution/state, auctions can also reveal information about individual buyers' realized values. This form of information revelation can be a relevant consideration in repeated auctions \citep[e.g.,][]{ortega1968, engelbrecht1983} or auctions with aftermarkets \citep[e.g.,][]{bukhchandani1989, goeree2003, haile2003}, and many papers highlight forward-looking buyers' signaling incentives and their effects on the auctioneer's revenue. Motivated by the question of privacy protection, \cite{eilat2023} compare different auction formats in terms of the information that the winner's payment reveals about the winner's value, as measured by the amount of entropy reduction. They assume a known value distribution, so the first-price auction is most informative (i.e., least privacy-preserving) in their setting, because its winning payment perfectly reveals the winner's value (by inversion of the bidding function). They also show that $k$th-price auctions are less informative the higher $k$. While our exercise gives rise to the same rankings among these particular auctions, the source of the rankings, the environment, and the information being compared are very different. For example, in our setting, the value distribution is unknown and the DM observes all buyers' bids, so the second-price auction perfectly reveals the winner's value (which equals the winner's bid), whereas the first-price auction is less informative about the winner's value.

Finally, we build on the literature on comparisons of statistical experiments, by adopting \cites{lehmann1988} accuracy order to compare the informativeness of different auction formats. The fact that this order is less conservative than \cites{blackwell} order allows for the sharp result that the first-price auction always Lehmann-dominates all well-behaved standard auctions; in contrast, depending on the underlying value distributions, the first-price auction need not Blackwell-dominate all such auctions.\footnote{E.g., under some value distributions, the upper endpoint of the third-price bid distribution is strictly decreasing in $\theta$ (Appendix~\ref{app:MLRP-parametric}), in which case first-price bids do not Blackwell-dominate third-price bids. However, the first-price auction Blackwell-dominates all other auctions in the uniform example in Figure~\ref{fig:uniform}.}
At the same time, the Lehmann order is conservative enough to ensure that more accurate auctions lead to robust payoff improvements in the large class of monotone decision problems. These payoff implications of the Lehmann order have been extended to more general decision problems \citep[e.g.,][]{quah2009}, ambiguous priors \citep{li2020}, and multi-dimensional signals or actions \citep{di2021, kim2023}; in Online Appendix~\ref{app:general lehmann}, we provide a generalization to non-i.i.d.\ signals that we use in the extensions in Section~\ref{sec:extensions}. Several papers apply the Lehmann order to the study of auctions in other contexts \citep[e.g.,][]{persico2000, bergemann2002, bobkova2024}.

\section{Model}\label{sec:model}

 {\bf Environment.} A group of $n$ buyers ($i = 1, \ldots, n$) compete for one indivisible good in an auction. Buyers have quasi-linear preferences and privately observe their values $v_i \in \mathbb{R}_+$ for the good, which are drawn i.i.d.\ from a cumulative distribution function (cdf) $F_\theta$. The value distribution $F_\theta$ depends on a state of the world $\theta \in \Theta$ that is drawn from a prior distribution $q_0\in\Delta(\Theta)$, where $\Theta$ is a measurable subset of $\mathbb R$. Each distribution $F_\theta$ is supported on a bounded interval $I_\theta=[\underline v_\theta, \overline v_\theta) \subseteq \mathbb R_+$ and admits a density $f_\theta$ that is continuous and strictly positive on $(\underline v_\theta, \overline v_\theta)$. We assume that higher states correspond to higher values, in the sense that the value distributions $(F_\theta)$ satisfy the monotone likelihood-ratio property (MLRP), i.e., $f_\theta(v')f_{\theta'}(v)\leq f_\theta(v)f_{\theta'}(v')$ for all $v> v'$ and $\theta>\theta'$. Throughout the paper, we use increasing (resp.\ decreasing) to mean weakly increasing (resp.\ decreasing).
 
One natural interpretation of state $\theta$ is as a common-value component (e.g., market fundamentals or the quality of the good). Indeed, in the current setting, each buyer's value can be written as $v_i = u(\theta, \varepsilon_i)$ for some increasing function $u$ of the common-value component $\theta$ and an idiosyncratic component $\varepsilon_i$, where $\varepsilon_i$ is drawn independently across buyers and $\theta$.\footnote{For example, we can write $v_i = u(\theta, \varepsilon_i)$, where $\varepsilon_i$ is i.i.d.\ uniform on $[0, 1]$ and $u(\theta, \varepsilon) = F_\theta^{-1} (\varepsilon)$.} A special case is when $v_i=\theta+\varepsilon_i$, where (to ensure the MLRP) $\varepsilon_i$ is drawn from a state-independent distribution with a log-concave density.

As motivated in the Introduction, we are interested in settings where buyers are better informed about $\theta$, and hence about the value distribution $F_\theta$, than a DM (e.g., the auctioneer himself or some third party), who seeks to draw inferences about $\theta$ from buyers' bids in the auction. For now, we focus on the benchmark where buyers perfectly observe $\theta$---in addition to their values $v_i$ or, equivalently, their idiosyncratic components $\varepsilon_i$---while the DM does not observe $\theta$. Remark~\ref{rem:type-only} below further discusses the assumption that buyers observe both $v_i$ and $\theta$; Section~\ref{sec:imperfect} extends our main result to a setting where buyers observe noisy signals about $\theta$.

\medskip

\noindent {\bf Auctions.} We will compare a large class of auction formats in terms of the information that buyers' bids reveal about $\theta$. Specifically, we consider the following class of \textbf{\textit{(standard) auctions}} \citep[e.g.,][]{krishna2009}: Each buyer $i$ simultaneously submits a bid $b_i  \in\mathbb R_+$. Buyers are ranked in terms of the order statistics 
$b^{(1)} \geq b^{(2)} \geq ... \geq b^{(n)}$ of their bids; in case of ties, an arbitrary tie-breaking rule is used to strictly rank buyers. The highest bidder receives the good. For each 
$\ell=1,\ldots, n$, the $\ell$th-highest bidder pays an amount $\psi_ \ell (b_1, \ldots, b_n) \in \mathbb{R}_+$ to the auctioneer, which can depend on her rank $\ell$ and the entire profile of bids. We assume that each payment function $\psi_\ell : \mathbb{R}^n_+ \to \mathbb{R}_+$ is increasing and symmetric under permutations (i.e., does not depend on bidders' identities). Hence, we can treat $\psi_\ell$ as a function $\psi_\ell (b^{(1)},\ldots, b^{(n)})$ of realized order statistics. We also assume that for all $\ell$, $\psi_\ell (b^{(1)},\ldots, b^{(n)})=0$ whenever $b^{(\ell)}=0$. This will ensure that all buyers receive a nonnegative interim payoff in equilibrium.  We identify each auction with its profile $\psi = (\psi_\ell)_{\ell =1, \ldots, n}$ of payment functions.

Standard auctions nest several important auction formats. For example, under the \textit{\textbf{$k$th-price auction}} (for $k = 1, \ldots, n$), payment functions $(\psi_\ell)$ are given by 
\[
\psi_\ell \left(b^{(1)},\ldots, b^{(n)}\right)=
\begin{cases}
b^{(k)} &\text{ if } \ell=1, \\
0 &\text{ if } \ell\not=1;
\end{cases}
\]
that is, the winner pays the $k$th-highest bid, while all other bidders pay nothing. 
  Under the \textit{\textbf{all-pay auction}}, payment functions $(\psi_\ell)$ are given by
\[
\psi_\ell \left(b^{(1)},\ldots, b^{(n)} \right)=b^{(\ell)}  \quad \text{ for all } \ell = 1, \ldots n;
\]
that is, all buyers pay their own bid. Our results extend readily to auctions with reserve prices (Section~\ref{sec:reserve}) and stochastic payments (by identifying $\psi_{\ell}$ with the expected payment).

Since buyers observe the state, each buyer $i$'s strategy is a collection of bidding functions $b_{i,\theta} (v_i)$ indexed by $\theta$. As is common, we restrict attention to auctions $\psi$ that admit a symmetric monotone equilibrium at each state $\theta$. That is, at each $\theta$, buyers' equilibrium strategies take the form $b^{\psi}_{i,\theta}(v_i)=b^{\psi}_\theta(v_i)$ for all $i$, where $b^{\psi}_\theta:I_\theta\to\mathbb R_+$ is continuous and strictly increasing. When there is no risk of confusion, we use the term equilibrium to refer to a symmetric monotone equilibrium.
Henceforth, we fix one such equilibrium $(b^\psi_\theta)$ for each auction $\psi$; in case of multiplicity, our results will not depend on which equilibrium of $\psi$ we consider. Let $G^{\psi}_\theta$ denote the induced \textit{\textbf{bid distribution}} in state $\theta$, i.e., the cdf defined by
\[
G^{\psi}_\theta (b) = F_\theta ( (b^{\psi}_\theta )^{-1} (b) ), \text{ for all } b \in \mathbb{R}_+.
\]
Thus, $G^{\psi}_\theta (b)$ represents the probability that an individual buyer bids at most $b$ when her value is drawn from $F_\theta$ and she bids according to $b^{\psi}_\theta$.\footnote{Here $(b^{\psi}_\theta )^{-1} (b) := \inf\{v \in I_\theta : b^{\psi}_\theta (v)\geq b\}$ denotes the left inverse of $b^{\psi}_\theta$, with $\inf \emptyset = \overline v_\theta$. We sometimes use $b^{\psi}_\theta(\overline v_\theta)$ to denote $\lim_{v \to \overline v_\theta} b^{\psi}_\theta (v)$. \label{fn:left-inverse}}

Note that since buyers know $\theta$ and their values are i.i.d.\ conditional on $\theta$, they face an independent private value auction at each $\theta$. (Section~\ref{sec:extensions} considers interdependent and correlated values). Thus, standard revenue-equivalence arguments imply that (all symmetric monotone equilibria of) all auctions $\psi$ yield the same expected revenue for the auctioneer. Note also that, by focusing on standard auctions and their symmetric monotone equilibria, we rule out more general mechanisms or equilibria where buyers can report the state. This is because our goal in this paper is to shed light on commonly used auction formats, by analyzing how they shape what information a DM can extract from buyers' bids. In contrast, we do not seek to explore all tools a designer might use to incentivize buyers to reveal their information.\footnote{If a designer can augment any $\psi$-auction by asking buyers to additionally submit reports of $\theta$ (without using these reports to determine allocations/payments), there is an equilibrium in which all buyers report $\theta$ truthfully. However, this channel is unavailable to a third-party DM, who cannot control the mechanism used to sell the good. Moreover, even if the DM is the auctioneer himself, in applications, state $\theta$ might be a reduced-form representation of information (e.g., buyers' assessments of market fundamentals) that may be nontrivial to communicate directly. Thus, broadly in the spirit of classic work emphasizing prices and other low-dimensional market messages as a way to economize on communication and information-processing costs
\citep[e.g.,][]{hayek1945,green1987}, it is natural to ask what can be learned about $\theta$ solely from monetary bids in standard auctions. Our focus on symmetric monotone equilibria also rules out somewhat artificial asymmetric equilibria that reveal $\theta$ in certain auctions: e.g., second-price auctions admit equilibria where a fixed buyer $i$ submits bids $b_{\theta,i}>\overline v_\theta$ that are distinct across $\theta$, while all other buyers always bid $0$. \label{fn:report}}

\medskip			

\noindent{\bf Experiments and Lehmann order.} While all auctions in our setting are revenue-equivalent, they may differ in terms of the information that equilibrium bids reveal about $\theta$. To formalize this, we view each auction $\psi$ as inducing a statistical experiment:
Conditional on each state $\theta$, the auction produces $n$ i.i.d.\ signals about $\theta$, namely the $n$ buyers' bids, which correspond to $n$ independent draws from the equilibrium bid distribution $G^{\psi}_\theta$ in state $\theta$. Our goal is to compare the informativeness of these statistical experiments $(G^{\psi}_\theta)$ across different auction formats $\psi$.

To do so, we measure informativeness using \cites{lehmann1988} classic accuracy order over experiments, which is a widely used relaxation of \cites{blackwell} order in settings with one-dimensional signals and states. Applied to two bid distributions $G = (G_\theta)$ and $\hat G = (\hat G_\theta)$, we say that $G$ is \textbf{\textit{more accurate}} than $\hat G$ if for all states $\theta>\theta'$ and bids $b \in \mathbb{R}_+$,\footnote{Here $G^{-1}_\theta (x) := \sup \{b \in {\rm supp}\, G_\theta : G_\theta (b)\leq x\}$ denotes the right inverse of cdf $G_\theta$, with $\sup\emptyset=0$. \label{fn:right-inverse}}
\[
G^{-1}_\theta \left(\hat G_\theta \left( b \right) \right) \geq G^{-1}_{\theta'}\left(\hat G_{\theta'} \left(b \right) \right).
\]
That is, $G^{-1}_\theta \left(\hat G_\theta \left( b \right) \right)$ is increasing in $\theta$, capturing a sense in which the bid distributions $G_\theta$ are more sensitive to increases in the state $\theta$ (and hence more informative about $\theta$) than the bid distributions $\hat G_\theta$. Equivalently, if $G$ is more accurate than $\hat G$ and $G_{\theta'}(b)=\hat G_{\theta'}(\hat b)\in(0,1)$, then $G_\theta(b)\leq\hat G_\theta(\hat b)$ for all $\theta>\theta'$. Thus, as the state rises, $b$'s quantile
under $G$ decreases more (or increases less) than that of $\hat b$ under $\hat G$. 

The Blackwell order is a very conservative partial order: observing signals from a Blackwell-more informative experiment guarantees a DM higher payoffs in \emph{all} decision problems. The accuracy order is less conservative, but \cite{lehmann1988} shows that it retains robust payoff implications in a large class of decision problems: If $G$ is more accurate than $\hat G$, then observing signals from $G$ rather than $\hat G$ guarantees the DM higher payoffs in all monotone decision problems (provided $(\hat G_\theta)$ satisfies the MLRP); see Section~\ref{sec:payoff} for the formal statement. As we will see, such decision problems include several natural settings where a DM uses information gleaned from observing bids in an auction.

\begin{rem}[{\bf Joint observation of $v_i$ and $\theta$}]\label{rem:type-only}
A crucial feature of our model is that buyers observe two-dimensional information---their own values $v_i$ and a common state $\theta$ that influences the value distribution $F_\theta$ (Section~\ref{sec:imperfect} extends to the case where buyers observe a common noisy signal of $\theta$). As discussed above, a natural interpretation of this feature is that each buyer $i$ observes both a common-value component $\theta$ and an idiosyncratic component $\varepsilon_i$ that jointly determine her value $v_i = u(\theta, \varepsilon_i)$; for example, in the context of Treasury auctions, $\theta$ may represent market fundamentals and $\varepsilon_i$ an idiosyncratic liquidity shock.\footnote{Motivated by empirical evidence that bidders in Treasury auctions possess a substantial amount of common information, \cite{pycia2026} study a related model of (multi-unit) auctions, where buyers observe a common signal whose realization is unknown to the seller and (possibly) an additional private signal. They analyze the structure of revenue-maximizing pay-as-bid auctions.}

To understand the role of this feature for our analysis, suppose instead that buyers only observe their values $v_i$ but not $\theta$, and play a symmetric monotone equilibrium under the prior $q_0 \in \Delta(\Theta)$. Then observing buyers' bids still reveals information about $\theta$, but all auction formats $\psi$ are informationally equivalent. Indeed, in this setting, each buyer $i$'s realized bid $b^{\psi} (v_i)$ in each state $\theta$ depends on $\theta$ only through her value $v_i$. Thus, under any auction $\psi$, a DM who observes these bids can invert the bidding function to learn buyers' values: That is, $(G^{\psi}_\theta)$ is equivalent (in terms of accuracy) to sampling from the value distributions $(F_\theta)$. In contrast, in our setting, buyers' bids $b^{\psi}_\theta (v_i)$ may depend on both $v_i$ and $\theta$, so a DM who does not know $\theta$ cannot in general back out buyers' values from their bids. Hence, different auctions $\psi$ can differ in terms of the information that $(G^{\psi}_\theta)$ reveals about $\theta$. \finex
\end{rem}

\begin{rem}[{\bf Partial bid observation}]\label{rem:partial-observation} We have framed our exercise as if the DM observes all $n$ buyers' auction bids, which is especially natural when the DM is the auctioneer himself. However, the analysis goes through under various forms of partial bid observation---for example, if the DM only observes a subsample or selected statistics of bids, which can be relevant when the DM is a third party.\footnote{For instance, in Treasury auctions, an important class of DMs are dealers, i.e., designated financial institutions who field clients' bids and thus observe a subsample of all bids, which they can use to inform their own future trades \citep[e.g.,][]{boyarchenko2021}. Some energy and environmental market auctions publicly disclose selected statistics, such as maximum, minimum, and median bids \citep[e.g., for emissions-allowance auctions, see][]{potomac2026}, which can provide valuable information to third-party DMs such as energy-market analysts and investors.} Specifically, note that by comparing the equilibrium bid distributions $(G^\psi_\theta)$ in terms of their Lehmann accuracy, we are analyzing the informativeness of a random \emph{individual} buyer's bid about the state. By \cite{lehmann1988}, the fact that more accurate bid distributions guarantee the DM higher payoffs in all monotone decision problems then holds for any fixed number $k \leq n$ of bid observations (i.e., independent draws from $(G^\psi_\theta)$); see Corollary~\ref{cor:lehmann} for the formal statement. Moreover, as we explain in Section~\ref{sec:payoff}, the analysis also goes through if the DM observes only a fixed order statistic (e.g., the highest bid or median bid) or a fixed collection of order statistics (e.g., the $k$ highest bids). In contrast, as we discussed in Section~\ref{sec:literature}, allowing the auction format itself to determine which order statistics of bids the DM observes introduces an additional channel for informational differences across auctions that is not our focus in this paper.\footnote{However, some of our main insights remain relevant even under this additional channel. Specifically, consider observing only the winner's payment under the first- vs.\ second-price auction, i.e., the highest first-price bid $(b^{\rm 1st}_\theta)^{(1)}$ vs.\ the second-highest second-price bid $(b^{\rm 2nd}_\theta)^{(2)}$. Our analysis below implies that observing $(b^{\rm 1st}_\theta)^{(1)}$ is more accurate than observing the highest second-price bid $(b^{\rm 2nd}_\theta)^{(1)}$. Moreover, since second-price bids coincide with buyers' values, the accuracy comparison of $(b^{\rm 2nd}_\theta)^{(1)}$ and $(b^{\rm 2nd}_\theta)^{(2)}$ reduces to that of the value order statistics $v^{(1)}$ and $v^{(2)}$, and \cite{awaya2026} provide conditions on $(F_\theta)$ and $n$ under which $v^{(1)}$ is more accurate than $v^{(2)}$. Thus, under the latter conditions, our analysis combined with \cite{awaya2026} implies that, by transitivity of the accuracy order, the first-price winner's payment $(b^{\rm 1st}_\theta)^{(1)}$ is more accurate than the second-price winner's payment $(b^{\rm 2nd}_\theta)^{(2)}$.}
 \finex
\end{rem}

\section{Analysis}\label{sec:analysis}
\subsection{Accuracy and Strategic Responsiveness}

We now proceed to analyze the accuracy order over auctions. As a first step, the following lemma provides an alternative characterization of this order in terms of bidding functions.
The proof follows directly  from the definition. 

\begin{lem}\label{lem:competition}
Consider any two auctions with equilibrium bidding functions and bid distributions $(b_\theta), (G_\theta)$ and $(\hat b_\theta), (\hat G_\theta)$, respectively. The following are equivalent:
\begin{enumerate}
\item  $G = (G_\theta)$ is more accurate than $\hat G = (\hat G_\theta)$.  
\item For all states $\theta>\theta'$ and bids $b \in \mathbb{R}_+$, we have
\[b_\theta \left(\hat{b}_\theta^{-1} \left(b \right) \right)\geq b_{\theta'} \left(\hat{b}_{\theta'}^{-1} \left(b \right) \right).\]
\end{enumerate}
\end{lem}

To interpret Lemma~\ref{lem:competition}, note that increases in $\theta$ have two effects on the bid distribution. First, there is the direct effect that, by assumption, a higher state $\theta$ corresponds to a higher value distribution $F_\theta$; for a \emph{fixed} monotonic bidding function, this leads the corresponding bid distribution to increase with $\theta$. This effect is independent of the auction format. Second, however, there is a \textbf{\textit{strategic responsiveness effect}}: For any fixed value $v$, $v$'s equilibrium bid $b_\theta (v)$ may adjust to increases in $\theta$, as the corresponding increase in the value distribution $F_\theta$ means that $v$ faces a more competitive pool of other bidders. Crucially, the sign and magnitude of this adjustment depend on the auction format. Lemma~\ref{lem:competition} shows that $G$ is more accurate than $\hat G$ if and only if the strategic responsiveness effect is more positive under $b$ than $\hat b$, in the sense that, for all $v$, $b_\theta (v)$ increases more (or decreases less) with $\theta$ than does $\hat{b}_\theta(v)$. Intuitively, stronger strategic responsiveness amplifies the sensitivity of the bid distributions to the state, increasing accuracy.

Given Lemma~\ref{lem:competition}, comparing the accuracy of different auctions amounts to understanding which auctions induce a stronger strategic responsiveness effect.

\subsection{Canonical Auction Formats}\label{sec:canonical}

To build intuition, we first consider some canonical auctions whose equilibrium bidding functions admit closed-form solutions: the first-price, second-price, third-price, and all-pay auctions. The first-price, second-price, and all-pay auction always admit a unique symmetric monotone equilibrium, whose bid distributions we denote by $G^{\rm 1st} = \left( G^{\rm 1st}_\theta \right)$, $G^{\rm 2nd} = \left( G^{\rm 2nd}_\theta \right)$, and $G^{\rm all} = \left( G^{\rm all}_\theta \right)$. Denote by $G^{\rm 3rd} = \left( G^{\rm 3rd}_\theta \right)$ the bid distributions under the unique symmetric monotone equilibrium of the third-price auction whenever it exists.

\begin{prop}\label{prop:example} \

\begin{enumerate}
\item $G^{\rm 1st}$ is more accurate than $G^{\rm 2nd}$, which in turn is more accurate than $G^{\rm 3rd}$. 
\item $G^{\rm 1st}$ is more accurate than $G^{\rm all}$. 
\end{enumerate}

\end{prop}

Proposition~\ref{prop:example} is obtained by applying Lemma~\ref{lem:competition} to the corresponding closed-form bidding functions. Under the second-price auction, the symmetric monotone equilibrium is for each buyer to bid her true value in all states, so $b^{\rm 2nd}_\theta (v) = v$ is \emph{constant} in $\theta$ for all $v$.  In contrast, under the first-price auction, the equilibrium bidding function is given by
\[
b^{\rm 1st}_\theta(v)=\frac{\int_0^{v} w  dF_\theta^{n-1}(w)}{F^{n-1}_\theta(v)}={\mathbb E}_\theta[v^{(2)}| v=v^{(1)}] \leq v,
\]
 where $v^{(1)},\ldots, v^{(n)}$ denote the order statistics of $n$ i.i.d.\ draws from $F_\theta$. Hence, by the MLRP of $(F_\theta)$, $b^{\rm 1st}_\theta(v)$ is \emph{increasing} in $\theta$ for all $v$. Thus, $G^{\rm 1st}$ is more accurate than $G^{\rm 2nd}$ by Lemma~\ref{lem:competition}. 
 
To interpret, note that under the first-price auction, a DM who does not know state $\theta$ cannot in general perfectly back out buyers' values from their bids; in contrast, second-price bids directly reveal buyers' values to the DM. Thus, one might naively expect the second-price auction to be preferable for a DM who seeks to learn the value distribution $F_\theta$. However, this naive intuition neglects the role of the strategic responsiveness effect: While this effect is absent under the second-price auction, under the first-price auction it leads each buyer type to increase her bid with the state. This makes observing first-price bids even more informative of the state/value distribution than directly observing buyers' values.\footnote{In contrast, the comparison is not driven by within-state bid dispersion. If $F_\theta$ is log-concave, then first-price bids are less dispersed than second-price bids, i.e., $b^{\rm 1st}_\theta(v') - b^{\rm 1st}_\theta(v) \leq v' - v$ for all $v' > v$ \citep[e.g.,][Theorem~5]{bagnoli2005}. However, under a variant of the first-price auction where winning payments are scaled down by a constant (i.e., $\psi_1 = \frac{1}{c}b^{(1)}$ for some $c > 1$), equilibrium bids $b_\theta (v) = c b^{\rm 1st}_\theta (v)$ are scaled up accordingly. These scaled bids are informationally equivalent to first-price bids, but for any $v' > v$, their dispersion $c(b^{\rm 1st}_\theta(v') - b^{\rm 1st}_\theta(v))$ can exceed the second-price dispersion for large enough $c$.}

   \begin{figure}[t]
\begin{center}
\begin{subfigure}{0.21\textwidth}
  \centering
  \includegraphics[width=\linewidth]{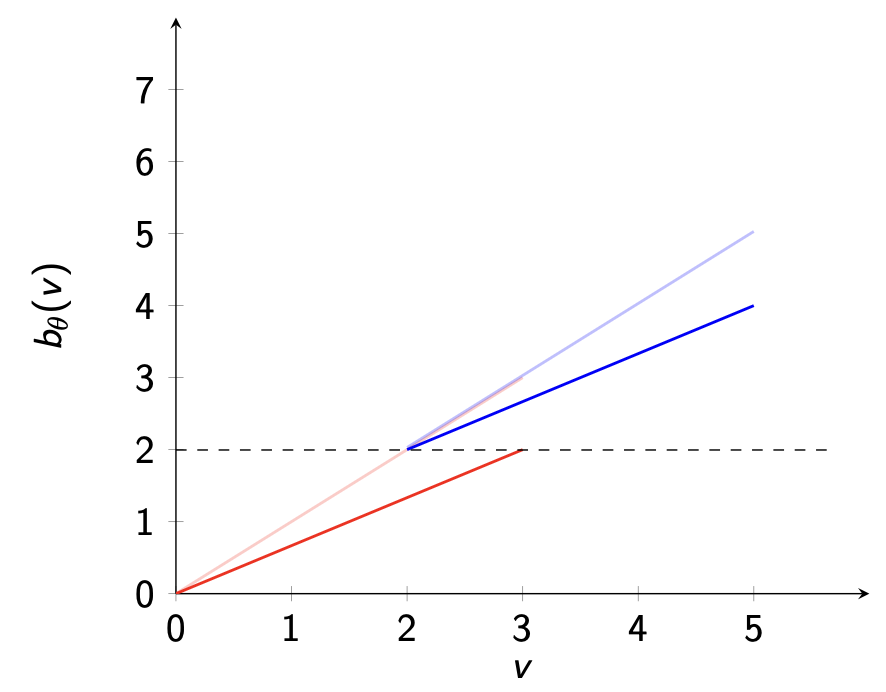}
  \caption*{ \tiny$b_{0}^{{\rm 1st}}(v)=\frac{2}{3}v\in [0,2]$, $b_{2}^{{\rm 1st}}(v)=\frac{2}{3}v+\frac{2}{3} \in [2, 4]$}
\end{subfigure} \quad
\begin{subfigure}{0.21\textwidth}
  \centering
  \includegraphics[width=\linewidth]{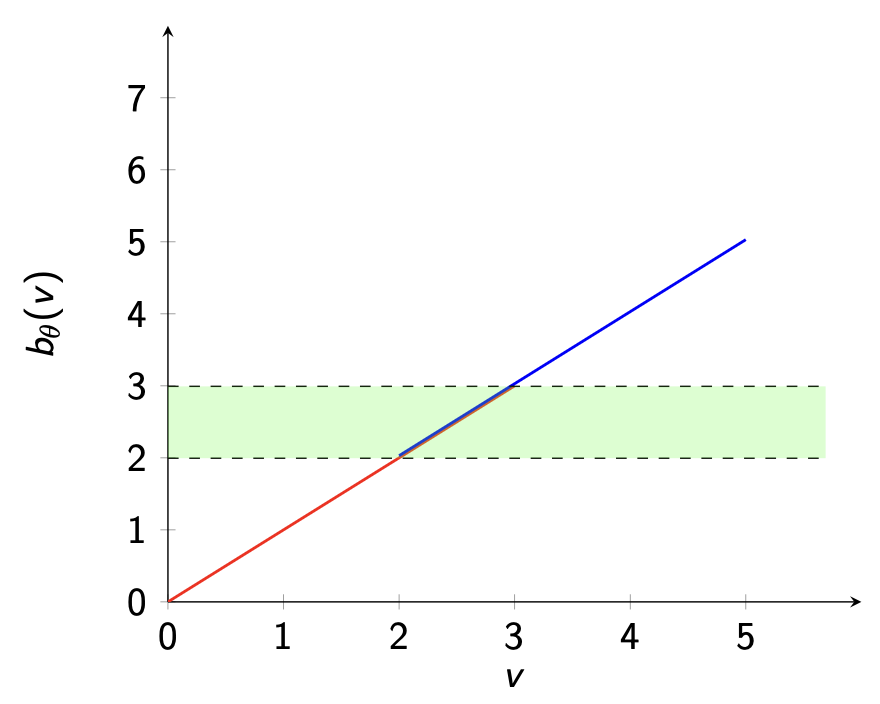}
   \caption*{ \tiny $b_0^{{\rm 2nd}}(v)=v \in [0, 3]$, $b_{2}^{{\rm 2nd}}(v)=v \in [2, 5]$}
\end{subfigure} \quad
\begin{subfigure}{0.21\textwidth}
  \centering
  \includegraphics[width=\linewidth]{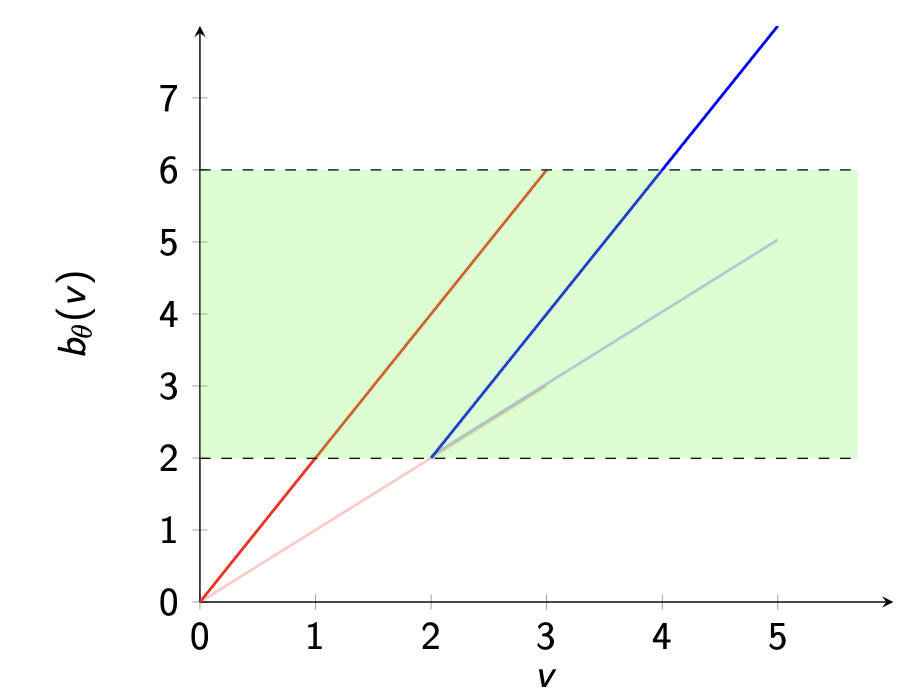}
   \caption*{ \tiny $b_0^{{\rm 3rd}}(v)=2v \in [0, 6]$, $b_{2}^{{\rm 3rd}}(v)=2v - 2\in [2, 8]$}
\end{subfigure} \quad
\begin{subfigure}{0.26\textwidth}
  \centering 
  \includegraphics[width=\linewidth]{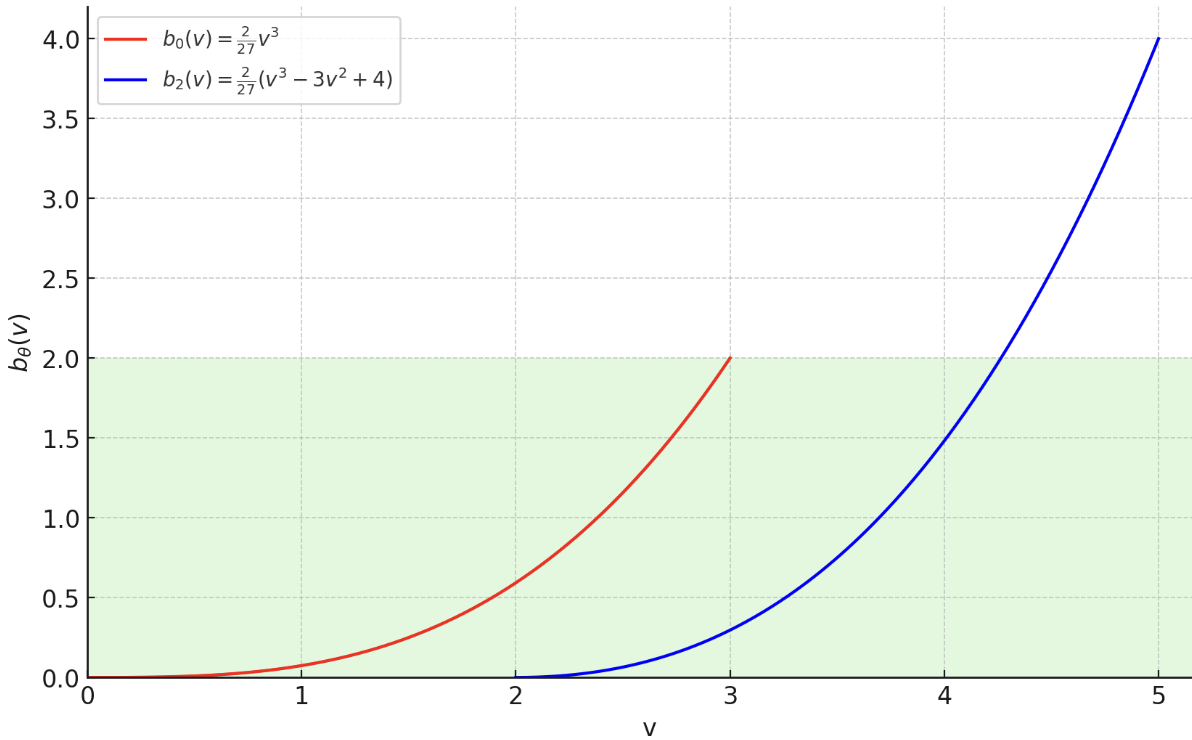}
   \caption*{ {\tiny $b_0^{\rm all}(v)=\frac{2}{27}v^3 \in [0, 2]$, \\$b_{2}^{{\rm all}}(v)=\frac{2}{27} (v^3 - 3v^2 + 4) \in [0,4]$}}
\end{subfigure}
\end{center}

\vspace{-4mm}
\begin{caption}
{\label{fig:uniform} \footnotesize Suppose $n =3$ and $F_\theta$ is the uniform distribution on $[\theta, \theta+3]$ where $\theta \in \Theta=\{0, 2\}$. The bidding functions under first-price, second-price, third-price, and all-pay auctions are shown from left to right in state $0$ (red) and state $2$ (blue).}
\end{caption}
\end{figure}

Figure~\ref{fig:uniform} provides a stark illustration of this comparison when $n =3$ and $F_\theta$ is given by the uniform distribution over $[\theta, \theta+3]$ for $\theta\in \Theta=\{0, 2\}$. Observe that the first-price bid distributions across the two states have non-overlapping supports (except at the single point $b =2$), so $(G^{\rm 1st}_\theta)$ is a \emph{fully informative} experiment. In contrast, under the second-price auction, the bids of types $v\in [2,3]$ provide no information, as illustrated by the green region in the figure, while other types' bids are fully informative. Hence, $G^{\rm 1st}$ strictly Blackwell-dominates $G^{\rm 2nd}$.

Under the third-price auction, the strategic responsiveness effect pushes in the opposite direction. There is a unique symmetric monotone equilibrium if and only if $v + \frac{F_\theta(v)}{(n-2)f_\theta(v)}$ is strictly increasing in $v$ for all $v$;\footnote{A sufficient condition is that each $F_\theta$ is log-concave \citep[e.g.,][]{krishna2009}.} in this case, the equilibrium bidding function is
\[
b^{\rm 3rd}_\theta(v)=v + \frac{F_\theta(v)}{(n-2)f_\theta(v)} \geq v.
\]
By the MLRP of $(F_\theta)$, the reverse hazard rates $\frac{f_\theta (v)}{F_\theta (v)}$ are increasing in $\theta$ \citep[e.g.,][Theorem 1.C.1]{shaked2007}. Hence, 
$b^{\rm 3rd}_\theta (v)$ is \emph{decreasing} in $\theta$ for all $v$. Thus, strategic responsiveness is negative, so $G^{\rm 3rd}$  is less accurate than $G^{\rm 2nd}$ by Lemma~\ref{lem:competition}. While for fixed $v$, third-price bids $ b^{\rm 3rd}_\theta(v)$ display more across-state variation than the constant second-price bids $b^{\rm 2nd}_\theta (v) = v$, this variation shifts the bid distribution $G^{\rm 3rd}_\theta$ in the opposite direction of changes in $F_\theta$, which leads $G^{\rm 3rd}_\theta$ to vary less with $\theta$ than $G^{\rm 2nd}_\theta$. This is again illustrated by the uniform example in Figure~\ref{fig:uniform}: Under the third-price auction, the set of values corresponding to uninformative bids is $[1, 4]$, which is larger than the interval $[2, 3]$ under the second-price auction. Thus, $G^{\rm 3rd}$ is a Blackwell garbling of $G^{\rm 2nd}$.

Finally, equilibrium bids under the all-pay auction are given by
\begin{eqnarray*}
b^{\rm all}_\theta(v)&=&\int_{0}^v wdF^{n-1}_\theta(w) = b^{\rm 1st}_\theta(v) F^{n-1}_\theta(v).
\end{eqnarray*}
That is, each $v$'s bid equals her first-price bid scaled by an adjustment term---$v$'s winning probability $F^{n-1}_\theta(v)$---that is decreasing in $\theta$ for all $v$. Intuitively, the fact that even non-winning bidders have to pay depresses all-pay bids relative to first-price bids, and more so in higher states $\theta$ where the pool of other bidders is more competitive. By Lemma~\ref{lem:competition}, this diminishes the informational content of all-pay bids relative to the first-price auction.\footnote{See the proof of Proposition~\ref{prop:own bid} for the details (Appendix~\ref{app:own bid}).} At the same time, depending on $(F_\theta)$, $b^{\rm all}_\theta (v)$ can be increasing in $\theta$ for large $v$ but decreasing in $\theta$ for small $v$ \citep[e.g., ][]{hopkins2007}. As a result, $G^{\rm all}$ cannot in general be ranked against $G^{\rm 2nd}$, where bids are constant in $\theta$. In Figure~\ref{fig:uniform}, the all-pay bids at values $v \in [0, 4.26]$ are partially informative (while other types' bids are fully revealing), so $G^{\rm all}$ is again a strict Blackwell garbling of $G^{\rm 1st}$. 
Note that, in this example, all-pay bids are decreasing in $\theta$ at all $v \in [2, 3]$ and that, by Lemma~\ref{lem:competition}, $G^{\rm all}$ is less accurate than $G^{\rm 2nd}$.

\begin{rem}[{\bf Open-format auctions}]\label{rem:open-format} Although our analysis focuses on sealed-bid auctions, the two most common
open formats are also no more informative than the first-price auction. In the standard Dutch (descending) auction equilibrium, only the winning bid
$b_\theta^{\rm 1st}(v^{(1)})$ is observed, which is a Blackwell garbling of observing all
first-price bids. In the standard English (ascending) auction equilibrium, the full
sequence of dropout prices reveals $(v^{(2)},\ldots,v^{(n)})$, which is a Blackwell
garbling of observing all second-price bids. \finex
\end{rem}

\subsection{Main Result}\label{sec:main}

We now consider the entire class of standard auctions, with general payment functions $\psi = (\psi_\ell)$ as defined in Section~\ref{sec:model}.
Among such auctions, the first-price auction has two defining features.\footnote{Both features are shared by monotone transformations of the first-price auction (i.e., where $\psi_1$ is some strictly increasing function of $b^{(1)}$), but the bid distributions of such $\psi$ are informationally equivalent to $G^{\rm 1st}$.}
 First, \textbf{\textit{only the winner pays}}: $\psi_\ell(b^{(1)},\ldots, b^{(n)})=0$ for all $\ell\not=1$. 
This feature is shared by all $k$th-price auctions, but not by the all-pay auction. Second, \textbf{\textit{payments depend only on own bids}} holding fixed a buyer's rank: For all $\ell$, the $\ell$th-highest bidder's payment $\psi_\ell(b^{(1)},\ldots, b^{(n)})$ depends only on $b^{(\ell)}$ and not on $b^{(k)}$ for $k\not=\ell$. This feature is shared by the all-pay auction, but not by $k$th-price auctions for $k \neq 1$.

Our main result is that these two features make the first-price auction most accurate among all well-behaved standard auctions:

\begin{thm}\label{thm:main}
For any standard auction $\psi$, $G^{\rm 1st}$ is more accurate than $G^{\psi}$, provided that $(G^{\psi}_\theta)$ satisfies the MLRP.
\end{thm}

Moving beyond the canonical auction formats in Proposition~\ref{prop:example}, Theorem~\ref{thm:main} compares the first-price auction against any standard auction $\psi$ that admits a symmetric monotone equilibrium. (As noted above, in case of multiplicity, the result applies to the bid distributions $G^{\psi}$ under any symmetric monotone equilibrium). Unlike the canonical auctions considered in Proposition~\ref{prop:example}, general auctions $\psi$ need not have closed-form equilibrium bidding functions. To establish a general ranking despite this difficulty, Theorem~\ref{thm:main} imposes some additional structure by requiring that the bid distributions $(G^{\psi}_\theta)$ satisfy the MLRP. Below, we first sketch the argument behind Theorem~\ref{thm:main} (see Appendix~\ref{app:main} for details), highlighting how the two defining features of the first-price auction jointly imply that it generates a stronger strategic responsiveness effect than any other such auction. Remark~\ref{rem:MLRP} then further discusses the MLRP assumption on $(G^{\psi}_\theta)$.

\medskip
\noindent {\bf Illustration of Theorem 1.} While $b^{\psi}_\theta$ need not admit a closed-form solution, the proof of Theorem~\ref{thm:main} invokes revenue equivalence to relate $b^{\psi}_\theta$ and $b^{\rm 1st}_\theta$ and to apply Lemma~\ref{lem:competition} to them. Rather than working with buyers' values, it is convenient to express bids as functions of buyers' value quantiles $x = F_\theta (v)$, which (regardless of the state) constitute $n$ i.i.d.\ draws from $U[0, 1]$. Let $x^{(1)}\geq x^{(2)}\geq \ldots \geq x^{(n)}$ denote the quantile order statistics, and let $p_i(x)$ denote the probability that a realized quantile $x$ is $i$th-highest. Then by standard revenue-equivalence arguments, in each state $\theta$, quantile $x$'s expected payment under the first-price auction is equal to $x$'s expected  payment under $\psi$, i.e.,
\[
\underbrace{p_{1}(x) b^{\rm 1st}_\theta(x)}_{\text{$x$'s expected payment under first-price}}  = \underbrace{\sum_{i=1}^n p_{i}(x){\mathbb E} \left[ \psi_i \left(b_\theta^\psi(x^{(1)}),\ldots, b_\theta^\psi(x^{(n)}) \right) \mid x^{(i)}=x  \right]}_{\text{$x$'s expected payment under $\psi$}}.
\]
Note that under $\psi$, $x$'s expected payment considers each possible rank $i$ of $x$ (weighted by its probability $p_i (x)$) and conditional on this rank, considers $i$'s expected payment $\psi_i$, which may depend on all buyers' bids. Defining $x^b_\theta =  (b^{\psi}_\theta)^{-1} (b)$ to be the quantile corresponding to each bid $b$ under $\psi$, for any interior bid $b$ we can rearrange to obtain 
\begin{equation}\label{eq:proof sketch}
b^{\rm 1st}_\theta \left((b_\theta^\psi)^{-1}(b)\right)= \sum_{i=1}^n \bl{\frac{p_{i}(x^b_\theta)}{p_1(x^b_\theta) }} \red{{\mathbb E} \left[ \psi_i \left(b_\theta^\psi(x^{(1)}),\ldots, b_\theta^\psi(x^{(n)}) \right) \mid  b^\psi_\theta(x^{(i)})=b  \right]}.
   \end{equation}

Thus, to prove that $G^{\rm 1st}$ is more accurate than $G^{\psi}$, it suffices by Lemma~\ref{lem:competition} to show that the right-hand side of (\ref{eq:proof sketch}) is increasing in $\theta$, which captures that the 1st-price auction induces a stronger strategic responsiveness effect than $\psi$. To do so, we prove that both the blue and red terms in (\ref{eq:proof sketch}) are increasing in $\theta$, where each of these terms highlights the role of one of the two defining features of first-price auctions.

First, consider the blue terms $\frac{p_{i}(x^b_\theta)}{p_1(x^b_\theta) }$, i.e., the relative probability that the quantile $x^b_\theta$ corresponding to bid $b$ under $\psi$ is $i$th-highest vs.\ highest. For any rank $i$, this is increasing in $\theta$: First, any fixed bid $b$'s quantile $x^b_\theta$ is decreasing in $\theta$ under $\psi$, because the bid distributions $G^{\psi}_\theta$ are FOSD-increasing in $\theta$ (by the MLRP of $(G^{\psi}_\theta)$); second, $p_i(x)/p_1(x)$ is decreasing in $x$ for every $i$, by a
standard property of rank probabilities in an i.i.d.\ sample.

The blue terms capture the informational advantage of the first-price auction over auctions that violate its first defining feature---that only the winner pays. Indeed, for any auction $\psi$ that shares this feature (e.g., $k$th-price auctions), the blue terms are irrelevant: they do not affect the right-hand side of (\ref{eq:proof sketch}) as $\psi_i \equiv 0$ for all $i \neq 1$.

 Second, consider the red terms ${\mathbb E} \left[ \psi_i \left(b_\theta^\psi(x^{(1)}),\ldots, b_\theta^\psi(x^{(n)}) \right) \mid  b^{\psi}_\theta(x^{(i)})=b  \right]$, i.e., the $i$th-highest bidder's expected payment when her bid is $b$. A complication is that $\psi_i$ may depend on all other buyers' bids, which are in general correlated conditional on bid $b$ being $i$th-highest. However, observe that, by the Markov property of order statistics, the $i-1$ highest bids $\left(b_\theta^\psi(x^{(1)}),\ldots, b_\theta^\psi(x^{(i-1)}) \right)$ are independent of the $n-i$ lowest bids $\left(b_\theta^\psi(x^{(i+1)}),\ldots, b_\theta^\psi(x^{(n)} )\right)$ conditional on bid $b$ being $i$th-highest. Conditional on $b^{(i)}=b$, the former have the same distribution as the order statistics of $i-1$ i.i.d.\ draws from the bid distribution $G^{\psi}_\theta$ truncated below at $b$, while the latter have the same distribution as the order statistics of $n-i$ i.i.d.\ draws from $G^{\psi}_\theta$ truncated above at $b$. Moreover, by the MLRP of $(G^{\psi}_\theta)$,  both these truncated distributions are FOSD-increasing in $\theta$ for all $b$. Since the payment function $\psi_i$ is increasing, this implies that its expectation conditional on $b^{(i)} = b$ is increasing in $\theta$.
 
%

The red terms capture the informational advantage of the first-price auction over auctions that violate its second key feature---that all payment functions $\psi_{i}$ depend only on own bids $b^{(i)}$. Indeed, for any other auction $\psi$ that shares this feature (e.g., all-pay), the red terms reduce to $\psi_i (b)$ and hence are constant in $\theta$.

\begin{rem}[{\bf MLRP of $G^{\psi}$}]\label{rem:MLRP}
As we illustrated above, the proof of Theorem~\ref{thm:main} uses the MLRP assumption on $(G^{\psi}_\theta)$ to show that the red terms in (\ref{eq:proof sketch}) (i.e., the $i$th-highest bidder's expected payment conditional on her bid taking any value $b$) are always increasing in $\theta$.\footnote{We also used the MLRP assumption to show that the blue terms are increasing in $\theta$, but this step only required the less demanding condition that $(G^{\psi}_\theta)$ is FOSD-increasing in $\theta$.} However, we emphasize that the MLRP assumption on $(G^{\psi}_\theta)$ is not needed when comparing the first-price auction against many specific auctions $\psi$. In particular, we did not impose this assumption for the comparison against the canonical auctions in Proposition~\ref{prop:example}. It can also be dropped for any auctions $\psi$ where each $\psi_i$ depends only on $b^{(i)}$ (e.g., the class of auctions in Proposition~\ref{prop:own bid} below), and can be relaxed to reverse-hazard rate dominance when each $\psi_i$ depends only on lower-ranked bids $b^{(i)},\ldots, b^{(n)}$ (e.g., all $k$th-price auctions). 

At the same time, the MLRP of $(G^{\psi}_\theta)$ is implied by the MLRP of the value distributions $(F_\theta)$ in various important settings. This is immediate for second-price auctions, where $G^{\rm 2nd}_\theta = F_\theta$. In addition, Appendix~\ref{app:MLRP-parametric} exhibits a natural parametric family of value distributions where the MLRP of $(G^{\psi}_\theta)$ holds for several canonical auctions $\psi$. Moreover, Appendix~\ref{app:MLRP-large-n} shows that, under general value distributions $(F_\theta)$ satisfying some technical regularity conditions, the bid distributions under any $k$th-price auction satisfy the MLRP for all sufficiently large numbers of bidders.\footnote{While Appendix~\ref{app:MLRP-parametric} also exhibits examples of auctions $\psi$ for which $(G^{\psi}_\theta)$ violates the MLRP, it is worth noting that, even in these cases, $G^{\psi}$ cannot be strictly more accurate than $G^{\rm 1st}$. Indeed, Online Appendix~\ref{app:undominated} shows that for any value distributions $(F_\theta)$ and any standard auction $\psi$ that admits a unique symmetric monotone equilibrium, $G^{\psi}$ is never strictly more accurate than $G^{\rm 1st}$. Thus, subject to equilibrium uniqueness, $G^{\rm 1st}$ is undominated in terms of accuracy even when the MLRP assumption is dropped.}
\finex
\end{rem}


\subsection{Ranking Other Auctions}\label{sec:subclass}

To shed further light on the two defining features of the first-price auction, the following results establish accuracy rankings among two subclasses of auctions. Each subclass preserves one of the two features of the first-price auction, so the ranking within that subclass isolates the effect of violating the other feature.

First, consider $k$th-price auctions, which preserve the feature that only the winner pays: 
\begin{prop}\label{prop:k} For any $k = 1, \ldots, n-1$, the bid distribution $G^{k{\rm th}}$ under the $k$th-price auction is more accurate than $G^{(k+1){\rm th}}$, provided that $(G^{(k+1){\rm th}}_\theta)$ satisfies the MLRP.  
\end{prop}

Proposition~\ref{prop:k} extends the ranking among $G^{1{\rm st}}$, $G^{2{\rm nd}}$, and  $G^{3{\rm rd}}$ in Proposition~\ref{prop:example}. While all $k$th-price auctions with $k \geq 2$ violate the feature that the winner's payment depends only on her own bid, the resulting loss of informativeness is more severe when $\psi_1$ is determined by a lower-ranked bid, i.e., when $k$ is larger. In particular, $G^{n{\rm th}}$ is least accurate within this class. To prove Proposition~\ref{prop:k} (Appendix~\ref{app:k}), we do not work with explicit bidding functions, which are difficult to analyze for $k > 3$. Instead, we exploit the recursion
\[
b^{k{\rm th}}_\theta (v) ={\mathbb E}_\theta[b^{(k+1){\rm th}}_\theta(v^{(k+1)})|v^{(k)}=v]
\]
that relates $k$th-price and $(k+1)$th-price bidding functions.


Second, consider \textit{\textbf{own-pay auctions}}, which preserve the first-price auction's feature that payments depend only on own bids: There exists a nonzero vector $\alpha \in \mathbb{R}^n_+$ such that $\psi_\ell (b^{(1)},\ldots, b^{(n)})=\alpha_\ell b^{(\ell)}$ for all $\ell$, i.e., the $\ell$th-highest bidder pays a multiple $\alpha_\ell$ of her own bid.\footnote{This class of auctions is also studied in the context of contests or rent-seeking \citep[e.g.,][]{baye2012}.} This class nests the first-price auction ($\alpha_1=1$ and $\alpha_\ell=0$ for all $\ell \neq 1$) and the all-pay auction ($\alpha_\ell=1$ for all $\ell$). The own-pay auction induced by any nonzero vector $\alpha$ admits a unique symmetric monotone equilibrium, whose bid distributions we denote by $G^\alpha$.

\begin{prop}\label{prop:own bid}
 Take any nonzero $\alpha, \alpha' \in \mathbb R_+^n$ such that $\alpha_j'\alpha_i\geq\alpha_i'\alpha_j$ for all $i,j$ with $i<j$. Then $G^\alpha$ is more accurate than $G^{\alpha'}$.  \end{prop}

Thus, $G^\alpha$ is more accurate if higher-ranked bidders pay a larger multiple of their bids relative to lower-ranked bidders---that is, the closer the auction comes to satisfying the first-price auction's feature that only the winner pays. 
Consequently, the least accurate own-pay auction is the \textit{\textbf{loser-pay}} auction \citep{riley1981,krishna2009}, where $\alpha_n=1$ and $\alpha_\ell=0$ for all $\ell\neq n$. Proposition~\ref{prop:own bid} does not require the MLRP of the dominated bid distributions, as the analogs of the red terms in (\ref{eq:proof sketch}) are always constant in $\theta$.


\begin{rem}[{\bf War of attrition}]\label{rem:WOA}
Whereas the above rankings isolate the effect of violating one of the features of first-price auctions, one can also analyze the effect of combining two violations. 
As a canonical example, consider the \textbf{\textit{war-of-attrition}} auction \citep[e.g.,][]{krishna1997}, whose payment functions are given by $\psi_1 (b^{(1)},\ldots, b^{(n)})= b^{(2)}$ and $\psi_\ell (b^{(1)},\ldots, b^{(n)}) = b^{(\ell)}$ for all $\ell \geq 2$. This auction combines two departures from the first-price auction: as in the second-price auction, the winner pays the second-highest rather than the highest bid; and as in the all-pay auction, all non-winning bidders pay their own bids. The following result shows that combining these two departures reduces accuracy more than either departure in isolation. We denote by $G^{\rm woa}$ the bid distributions in the unique symmetric monotone equilibrium of the war-of-attrition auction.
\begin{prop}\label{prop:WOA}
Both $G^{\rm 2nd}$ and $G^{\rm all}$ are more accurate than $G^{\rm woa}$. \finex
\end{prop}
\end{rem}

\section{Payoff Implications}\label{sec:payoff}

To illustrate the payoff implications of our finding that first-price auctions are most accurate, we now explicitly incorporate a second stage, where the DM uses the information about $\theta$ gleaned from buyers' auction bids in a decision problem. Specifically, suppose $\Theta$ is compact and that after observing buyers' bids, the DM faces a \textbf{\textit{monotone decision problem}} ${\cal D}=(A, u)$ \`a la \cite{karlin1956}. Here, $A\subseteq\mathbb R$ is an action set that is either finite or a compact interval, and $u: A\times\Theta\to\mathbb R$ is a continuous utility function that is single-peaked in $a$ at all $\theta$, with a maximizer $a^* (\theta) = \argmax_{a \in A} u(a, \theta)$ that is increasing in $\theta$.

Let $U({\cal D}, G^\psi)$ denote the DM's expected payoff to choosing optimally at ${\cal D}$ when prior to his choice, he observes the $n$ buyers' bids $b_1,\ldots, b_n$ drawn from the equilibrium bid distributions $G^\psi$ of auction $\psi$. That is, 
\[
U({\cal D}, G^\psi)=\max_{\sigma: {\mathbb R}^n_+ \to A}{\mathbb E}[u(\sigma(b_1,\ldots, b_n),\theta)],
\] 
where the expectation is with respect to the state $\theta$ (which is drawn from the prior $q_0 \in \Delta(\Theta)$) and the bids $b_1,\ldots, b_n$ (which are i.i.d.\ draws from $G_\theta^\psi$ conditional on each $\theta$). 

\begin{cor}\label{cor:lehmann}
Consider any monotone decision problem $\cal D$. Suppose that $G^\psi$ is more accurate than $G^{\psi'}$, and $(G^{\psi'}_\theta)$ satisfies the MLRP. Then $U({\cal D}, G^\psi)\geq U({\cal D}, G^{\psi'})$. In particular, $U ({\cal D}, G^{\rm 1st}) \geq U ({\cal D}, G^{\psi'})$.
\end{cor}

Thus, more accurate auctions provide a higher value of information to the DM, robustly across all monotone decision problems and prior beliefs. This result is an immediate implication of the characterization of the accuracy order in \cite{lehmann1988} (Theorem~7.1), combined with the fact that the $n$ bids are i.i.d.\  conditional on each state $\theta$.
The MLRP of the less accurate bids $(G^{\psi'}_\theta)$ is needed to apply \cites{lehmann1988} result: Different from its role in Theorem~\ref{thm:main}, it serves to guarantee that the DM's optimal strategy is monotone, which is needed to obtain that $U({\cal D}, G^\psi)\geq U({\cal D}, G^{\psi'})$.\footnote{While \cite{lehmann1988} imposed the MLRP on both $G^{\psi}$ and $G^{\psi'}$, his result only requires the MLRP of $G^{\psi'}$; see, e.g., \cite{quah2009}, who extended \cites{lehmann1988} result to the more general class of interval-dominance order utilities.} Together with Theorem~\ref{thm:main}, this result then implies that $U ({\cal D}, G^{\rm 1st}) \geq U ({\cal D}, G^{\psi'})$.

Note that \cites{lehmann1988} result applies under any fixed number of i.i.d.\ draws from each experiment. Thus, Corollary~\ref{cor:lehmann} remains valid when the DM only observes some subsample of $k \leq n$ bids in each auction. Moreover, since the accuracy order is preserved under taking order statistics, the payoff implications are also unchanged if the DM only observes the $k$th-highest bid in each auction (for some fixed $k = 1, \ldots, n$).\footnote{Indeed, if $(G^\psi_\theta)$ is more accurate than $(G^{\psi'}_\theta)$, then $(H\circ G^{\psi}_\theta)$ is more accurate than $(H\circ G^{\psi'}_\theta)$ for any strictly increasing bijection $H:[0,1]\to[0,1]$, since
$(H\circ G^{\psi}_\theta)^{-1}\circ(H\circ G^{\psi'}_\theta) = (G^{\psi}_\theta)^{-1}\circ G^{\psi'}_\theta$. The cdf of the $k$th-highest order statistic corresponds to applying the transformation $H_k (x)=\sum_{j=0}^{k-1}\binom{n}{j}(1-x)^j x^{n-j}$. Although an arbitrary monotone transformation $H$ need not preserve the MLRP of $(G^{\psi'}_\theta)$, the order-statistic transformation $H_k$ does preserve it \citep[Theorem~1.C.33]{shaked2007}.}
 Finally, Online Appendix~\ref{app:general lehmann} provides an extension of Corollary~\ref{cor:lehmann} to non-i.i.d.\ signals; this extension implies that more accurate auctions $\psi$ also induce higher payoffs in all monotone decision problems when the DM observes a collection of order statistics (e.g., the $k$ highest bids  $b^{(1)}, \ldots, b^{(k)}$).

 As discussed in the Introduction, Corollary~\ref{cor:lehmann} can be applied to a wide range of economic settings, where the DM may be the auctioneer himself or a third party. Below we consider two simple examples: Section~\ref{sec:prediction} considers statistical prediction problems and quantifies the informational advantage of first-price auctions. Section~\ref{sec:reserve} considers the auctioneer's problem of setting a reserve price.

\subsection{Example 1: Statistical Prediction Problems}\label{sec:prediction}

As a first example, consider the subclass of monotone decision problems where $A =\Theta$ is a compact interval, the prior $q_0 \in \Delta (\Theta)$ admits a continuous and strictly positive density on $\Theta$,  and $u(a, \theta) = - \ell (|a-\theta|)$ for some strictly increasing and continuous function $\ell:\mathbb R_+\to\mathbb R_+$ with $\ell(0)=0$.  That is, the DM faces a statistical prediction problem (e.g., a financial analyst's forecasting problem in the context of Treasury auctions) where his disutility of failing to match the state is captured by the loss function $\ell$.
Let
\[
L(G^\psi)= \min_{\sigma: {\mathbb R}^n_+ \to \Theta}{\mathbb E}[\ell(|\sigma(b_1,\ldots, b_n)-\theta|)]>0
\] 
denote the DM's expected minimal loss when prior to choosing an action, he observes the $n$ buyers' bids in auction $\psi$. How close 
$L(G^\psi)$ comes to $0$ provides a quantitative measure of how informative $G^{\psi}$ is about $\theta$. For example, when $\ell(d)=d^2$, $L(G^\psi)$ coincides with the expected value of the DM's posterior variance about $\theta$.

It is immediate from Corollary~\ref{cor:lehmann} that $L (G^{\rm 1st})$ is minimal among all auctions $\psi$ for which $G^\psi$ satisfies the MLRP, and in particular $L  (G^{\rm 1st}) \leq L(G^{\rm 2nd})$. While the gap between $L (G^{\rm 1st})$ and $L(G^{\rm 2nd})$ depends on $(F_\theta)$, the following corollary shows that the informational advantage of first-price over second-price auctions can be arbitrarily large. We index bid distributions $G^{\psi}_n$ by the number of buyers $n$, and denote by $L_{\emptyset} = \min_{\hat\theta\in\Theta}{\mathbb E}[\ell(|\hat\theta-\theta|)]$ the \textbf{\textit{no-information loss}}, i.e., the DM's expected minimal loss when choosing solely based on his prior $q_0$.

\begin{cor}\label{cor:prediction} Fix any loss function $\ell$, numbers of bidders $n,m \geq 2$, and $\varepsilon>0$. There exist value distributions $(F_\theta)$ such that $L (G_m^{\rm 1st}) \leq\varepsilon$ but $L (G_n^{\rm 2nd}) \geq L_{\emptyset} - \varepsilon$. Thus, 
\begin{equation*}
\sup_{(F_\theta)}\frac{L (G_n^{\rm 2nd})}{L (G_m^{\rm 1st})}=\infty.
\end{equation*}
\end{cor}

That is, under some value distributions $(F_\theta)$, $L (G_m^{\rm 1st})\approx 0$ while $L (G_n^{\rm 2nd}) \approx L_\emptyset$, i.e., first-price bids provide almost perfect information, while second-price bids provide almost no information about $\theta$. Moreover, this is true even if the number of bidders $n$ in the second-price auction is arbitrarily larger than the number of bidders $m$ in the first-price auction. The proof (Appendix~\ref{app:cor-prediction}) provides an explicit construction of such value distributions.\footnote{Figure~\ref{fig:uniform} provides an example where $L(G^{\rm 2nd}_n) > 0$ for all $n$ while $L(G^{\rm 1st}_m) =0$ for $m \in \{2, 3\}$, but this relies on the discreteness of $\Theta$. Moreover, the example does not establish that second-price bids can be almost uninformative under value distributions that make first-price bids almost fully revealing.}

In general, a larger number of bidders is always more informative under the second-price auction (i.e., $L(G^{\rm 2nd}_n)$ is decreasing in $n$) as this corresponds to observing more draws from the value distribution. However, $L (G_n^{\rm 1st})$ can be non-monotonic in $n$: More bidders can dampen the strategic responsiveness effect, making any individual bid observation less accurate, and this effect can more than offset the informational benefit of an increased number of bid observations.\footnote{E.g., in Figure~\ref{fig:uniform}, $L(G^{\rm 1st}_n) =0$ for $n \in \{2, 3\}$ but $L(G^{\rm 1st}_n)  > 0$ for all $n > 3$. }


\subsection{Example 2: Choice of Reserve Price}\label{sec:reserve}

As a second example, suppose that, as motivated in the Introduction, the DM is the auctioneer and his decision problem is to use bid observations to set a future reserve price.

To study this setting, we first extend the main analysis to allow for a reserve price $r \geq 0$: Given payment functions $\psi = (\psi_\ell)$ as in Section~\ref{sec:model}, the highest bidder wins the good if and only if $b^{(1)} \geq r$. The winner pays $\max \{ r, \psi_1(b^{(1)},\ldots, b^{(n)})\}$; all other bidders pay  $\psi_\ell (b^{(1)},\ldots, b^{(n)})$ depending on their rank $\ell$. We modify the notion of symmetric monotone equilibrium by requiring that, at each state $\theta$, all buyers' bidding strategies are described by a map $b_\theta: I_\theta \to [r, \infty)\cup \{0\}$ that is strictly increasing and continuous for all $v\geq r$ and satisfies $b_\theta (r) = r$ and $b_\theta (v) = 0$ for all $v < r$. Thus, we interpret bidding $0$ as not participating in the auction, and in equilibrium, the participating buyers are precisely those with values $v \geq r$.\footnote{Some auctions where non-winning bidders make non-zero payments do not admit such equilibria when $r >0$. For example, in the all-pay auction, buyers with values just above $r$ have an incentive to deviate to bidding less than $r$. Online Appendix~\ref{app:fee} replaces reserve prices with entry fees; here, symmetric monotone equilibria exist for a broader class of auctions, and we obtain analogous insights as below.} Let $G^{\psi, r}=(G^{\psi, r}_\theta)$  denote the corresponding equilibrium bid distributions. The first-price and second-price auctions admit unique symmetric monotone equilibria, where for all $v \geq r$,
\begin{equation*}
b_\theta^{{\rm 1st}, r}(v)=
{\mathbb E}_\theta[\max\{v^{(2)}, r\} \mid v=v^{(1)}],
\quad \quad \quad
b_\theta^{{\rm 2nd}, r}(v)=
v. \end{equation*}

The following result extends Theorem~\ref{thm:main}:

\begin{prop}\label{prop:reserve}
 For any $r\geq0$, $G^{{\rm 1st}, r}$ is more accurate than $G^{\psi, r}$, provided that $(G^{\psi, r}_\theta)$ satisfies the MLRP. Moreover, $G^{{\rm 1st}, r'}$ is more accurate than $G^{{\rm 1st}, r}$ whenever $0\leq r'<r$.\end{prop}

We now apply Proposition~\ref{prop:reserve} to a simple two-period setting, where the auctioneer sets a reserve price each period and uses period-1 bid observations to optimally choose the period-2 reserve price. Suppose that each period $t = 1, 2$, the auctioneer chooses an auction format $\psi^t$ and reserve price $r^t \in [0, \overline r]$ to sell one non-storable good to $n^t$ short-lived buyers. As in Section~\ref{sec:model}, each buyer's value is i.i.d.\ with distribution $F_\theta$. Here, $\theta$ is drawn once and for all at the start of period $1$ and is observed by all buyers. Crucially, the auctioneer observes period-1 bids prior to choosing $(\psi^2, r^2)$. The auctioneer's objective is to maximize his total expected revenue across the two periods.

Assume that each $F_\theta$ satisfies Myerson regularity, i.e., $v-\frac{1-F_\theta(v)}{f_\theta(v)}$ is strictly increasing in $v$.  Then conditional on any $\theta$, revenue-equivalence arguments imply that the auctioneer's expected one-shot revenue under (a symmetric monotone equilibrium of) any auction $\psi$ with reserve price $r \geq 0$ is 
\[u(r, \theta)= \int^{\infty}_r  \left(v-\frac{1-F_\theta(v)}{f_\theta(v)} \right)  dF^n_\theta (v),
\]
which belongs to the class of preferences in \cite{karlin1956}.\footnote{The function is single-peaked in $r$ since $v-\frac{1-F_\theta(v)}{f_\theta(v)}$ is strictly increasing. It is maximized by $r^*$ that solves $r^*-\frac{1-F_\theta(r^*)}{f_\theta(r^*)}=0$. 
By MLRP of $(F_\theta)$,  $\frac{1-F_\theta(r)}{f_\theta(r)}$ is increasing in $\theta$ \citep[][Theorem 1.C.1]{shaked2007}, which implies that the maximizer $r^*$ is also increasing in $\theta$.}
 Thus, the optimal choice of a period-2 reserve price $r^2$ is a monotone decision problem. By Corollary~\ref{cor:lehmann}, period-2 revenue is higher the more accurate the period-1 bid observations $G^{\psi^1, r^1}$. 
 
As a result, Proposition~\ref{prop:reserve} yields that (among all well-behaved standard auctions $\psi^1$), the auctioneer achieves the highest total revenue by using a first-price auction in period 1: Under any period-1 reserve price $r^1$, this yields the same period-1 revenue as any other auction $\psi^1$, but generates more accurate information about $\theta$ and hence a higher period-2 revenue. Having reduced the period-1 format choice to a first-price auction, one can also determine the optimal period-1 reserve price $r^1$. The latter trades off the information revealed by $G^{{\rm 1st}, r^1}$ (which, by Proposition~\ref{prop:reserve}, is less accurate the higher $r^1$) against maximizing period-1 revenue. Due to this exploration-exploitation tradeoff, $r^1$ depends on $(F_\theta)$ and the number of bidders in each period, but is always lower than the statically optimal reserve price. 

Finally, in assuming that buyers are short-lived, this stylized example isolates the channel of learning about $F_\theta$, while abstracting away from other forces that arise with long-lived buyers who bid in both periods. For example, if all buyers are long-lived but myopic, with values drawn once and for all from $F_\theta$ in period 1, the auctioneer would want to learn buyers' realized values rather than the value distribution $F_\theta$ and hence would optimally choose a second-price auction in period 1.\footnote{If buyers are forward-looking, they may additionally have incentives to distort period-1 bids to influence the auctioneer's and other bidders' inferences about their values. Some papers study such incentives in the context of reserve-price choice in second-price auctions \cite[e.g.,][]{hummel2018, kanoria2021}. Comparing how these incentives differ across auction formats is a worthwhile direction for future work.} We view the current example as an approximation of settings with a large enough inflow of new buyers in period 2, in which case learning about $F_\theta$ remains a relevant consideration for the auctioneer and provides a force in favor of using a first-price auction in period 1.

\section{Extensions}\label{sec:extensions}

\subsection{Interdependent Values}\label{sec:interdependent}

So far, we assumed that conditional on each state $\theta$, buyers face an independent private value environment. However, the analysis extends to the following interdependent value setting. Suppose each buyer $i$'s value for the good is given by $u(v_i, v_{-i})$, where $u: \mathbb{R}_+ \times \mathbb{R}^{n-1}_+ \to \mathbb{R}_+$ is continuous, strictly increasing in the first argument, and increasing and symmetric in the remaining arguments. We continue to assume that conditional on each state $\theta$, each $v_i$ is drawn i.i.d.\ from $F_\theta$, with the same assumptions on $(F_\theta)$ as in Section~\ref{sec:model}, and that each buyer $i$ observes $v_i$ and $\theta$. As in Section~\ref{sec:model}, we consider standard auctions $\psi$ that admit some symmetric and strictly monotone equilibrium with bid distributions and bidding functions $(G^{\psi}_\theta)$ and $(b^{\psi}_\theta)$. In particular, letting
\begin{equation}\label{eq:u-theta}
u_\theta(v_i):={\mathbb E}_\theta[u(v_i, v_{-i})| v_i=v^{(1)}=v^{(2)}],
\end{equation} 
the unique (symmetric monotone) equilibrium bidding functions under the first- and second-price auction are, respectively, given by \citep[][]{milgrom1982,mcadams2007}:
\[
b^{\rm 1st}_\theta(v)={\mathbb E}_\theta[u_\theta(v^{(2)})| v^{(1)}=v] \quad \text{ and } \quad 
b^{\rm 2nd}_\theta(v)= u_\theta(v).
\]
Note that, unlike under private values, even second-price bids $b^{\rm 2nd}_\theta(v)$ are increasing in $\theta$ for each $v$. 
This contrasts with the VCG mechanism \citep{dasgupta2000}, which admits a truthful equilibrium with bids $b^{\rm vcg}_\theta(v)=v$ whenever $u$ satisfies the single-crossing condition.\footnote{That is, $u$ is continuously differentiable with $\partial_{v_i} u(v_i, v_{-i})>\partial_{v_j} u(v_i, v_{-i})$ for all $(v_i, v_{-i})$ and $i\not=j$.}

The following result shows that the first-price auction remains most accurate under interdependent values. Moreover, in contrast to the private value setting, the fact that second-price bids $b^{\rm 2nd}_\theta (v)$ are increasing in $\theta$ makes them more accurate than truthful VCG bids.\footnote{As in Remark~3, the Dutch auction again reveals only
$b_\theta^{\rm 1st}(v^{(1)})$, which is a garbling of observing all
first-price auction bids, while observing the dropout prices in the Milgrom--Weber equilibrium of the English
auction reveals $(v^{(2)},\ldots,v^{(n)})$, which is a garbling of observing all truthful
VCG bids. The latter holds because after every equilibrium dropout history, $i$'s stopping threshold is a $\theta$-independent, strictly increasing function of $v_i$, even though it need not equal $v_i$ \citep{milgrom1982}.}

\begin{prop}\label{prop:interdependent} \
\begin{enumerate}
\item For any standard auction $\psi$, $G^{\rm 1st}$ is more accurate than $G^{\psi}$, provided that $(G^{\psi}_\theta)$ satisfies the MLRP.
\item $G^{\rm 2nd}$ is more accurate than $G^{\rm vcg}$.
\end{enumerate}
\end{prop}

To prove the first part of Proposition~\ref{prop:interdependent}, we exploit analogous revenue-equivalence arguments as in Theorem~\ref{thm:main}. These arguments extend to the current setting as we maintain the assumption that types $v_i$ are i.i.d.\ conditional on each state $\theta$.

\subsection{Imperfect State Observation}\label{sec:imperfect}

We can further generalize Theorem~\ref{thm:main} to a setting where, instead of perfectly observing the state $\theta$, buyers observe noisy signals about $\theta$. 
Suppose that all buyers $i$ observe both a common public signal $\xi \in \mathbb{R}$ (that is unobserved by the DM) and a private signal $v_i \in \mathbb{R}_+$, and $i$'s value for the good is given by $u(\xi, v_i, v_{-i})$, where $u: \mathbb{R} \times \mathbb{R}_+ \times \mathbb{R}^{n-1}_+ \to \mathbb{R}_+$ is continuous, strictly increasing in the second argument, and increasing and symmetric in the last $n-1$ arguments. We assume that the joint distribution of $(\theta, \xi, v_1, \ldots, v_n)$ admits a density that is symmetric in $(v_1,\ldots, v_n)$ and that conditional on $\xi$, the private signals $v_i$ are i.i.d.\ across buyers and admit a continuous and strictly positive density supported on an interval $I_\xi = [\underline v_\xi, \overline v_\xi)$. This setting nests the interdependent value environment in Section~\ref{sec:interdependent} when $\xi = \theta$ and $u$ does not depend on $\xi$. It also nests a pure common value setting where $u(\xi, v_i, v_{-i}) = \mathbb{E} \left[\theta \mid \xi, v_i, v_{-i} \right]$.\footnote{Assuming that $\xi$ is public and $v_i$ are i.i.d.\ conditional on $\xi$ is useful as this allows us to exploit revenue-equivalence arguments. If buyers observe two-dimensional private signals $(v_i, \xi_i)$, equilibrium is difficult to characterize even under simple formats (e.g., first-/second-price) and need not exist \citep[e.g.,][]{pesendorfer2000, jackson2009}.}

Focusing again on standard auctions $\psi$ that admit symmetric and monotone (in $v_i$) equilibria, each buyer $i$'s bid can be written as $b_i=b_\xi^\psi(v_i)$. Note that, in contrast to the analysis so far, the $n$ buyers' bids $(b_1,\ldots, b_n)$ are exchangeable but need not be independent conditional on $\theta$.  Let $G^\psi=(G^\psi_\theta)$ denote the distribution of an individual bid $b_i$ conditional on each state $\theta$.

The following result extends Theorem~\ref{thm:main} and Proposition~\ref{prop:interdependent} under an affiliation condition:

\begin{prop}\label{prop:imperfect}
Suppose $(\theta, \xi, b_1,\ldots, b_n)$ is affiliated under both the first-price auction and some auction $\psi$.
 Then $G^{\rm 1st}$ is more accurate than $G^{\psi}.$
\end{prop}

The affiliation condition for auction $\psi$ generalizes the MLRP assumption on $(G^{\psi}_\theta)$ in our previous results to accommodate correlated bids. 
To prove Proposition~\ref{prop:imperfect} (Appendix~\ref{app:imperfect}), we extend Lemma~\ref{lem:competition} to incorporate imperfect observation of the state. To do so, we additionally impose an affiliation assumption on the first-price auction, which had no counterpart in our previous results.\footnote{A simple example where affiliation holds is if (i) $(v_i)$ and $\theta$ are independent conditional on $\xi$ and $(\theta,\xi)$ is affiliated,  and (ii) the bid distribution conditional on $\xi$ satisfies the MLRP. Note that, in this case, (ii) is equivalent to the MLRP of bids in an (interdependent or private) i.i.d.\ value environment where $\xi$ is viewed as the state. See Appendix~\ref{app:MLRP} for natural conditions that ensure the MLRP of bids in the private-value setting. 
}

While it still follows from Proposition~\ref{prop:imperfect} and \cite{lehmann1988} that the first-price auction induces higher payoffs than $\psi$ in all monotone decision problems if a single bid is observed, this is not immediate under multiple bid observations because bids may be correlated conditional on $\theta$. However, Online Appendix~\ref{app:general lehmann} shows that the multiple-bid payoff comparison extends under concave utility functions if $(v_i)$ and $\theta$ are independent conditional on $\xi$.

\subsection{Correlated Values and Multi-Unit Auctions}\label{sec:other-extensions}

Finally, we consider two extensions that are further removed from the baseline model. Accordingly, we focus on extending the comparison between first- and second-price auctions.

First, we revisit the setting in Section~\ref{sec:interdependent} but allow $(v_1,\ldots, v_n)$ to be correlated conditional on $\theta$.  We additionally assume that $u$ is continuously differentiable and $\partial_1 u \geq \kappa>0$ for some $\kappa$. Generalizing the MLRP of $(F_\theta)$, we assume that the joint distribution of $(\theta, v_1,\ldots,v_n)$ admits a  density and is affiliated. Moreover, conditional on $\theta$, $(v_1,\ldots,v_n)$ is exchangeable and admits a continuously differentiable joint density that is strictly positive on its support $I_\theta^n$, where $I_\theta=[\underline v_\theta, \overline v_\theta)\subseteq \mathbb{R}_+$. Let ${F}_\theta(v' \mid v)$ denote the cdf of $v' = \max_{j \neq i} {v}_j$ conditional on ${v}_i = v$ and $\theta$, with corresponding density ${f}_\theta(v' \mid v)$.  The first- and second-price auctions admit unique symmetric monotone equilibria:
\[
b_\theta^{\rm 2nd}(v)=u_\theta(v)  \quad \text{ and } \quad   b_{\theta}^{\rm 1st}({v}) = \int_0^{v} u_\theta(v') d L_\theta (v' \mid v),
\]
where $u_\theta$ is as defined in (\ref{eq:u-theta}) and $L_\theta( v' \mid v)$ is a cdf on $[0,v]$ defined by $L_\theta( v' \mid v) =e^{- \intop_{v'}^{{v}} \frac{{f}_\theta({t} \mid {t})}{{F}_\theta({t} \mid {t})} d {t}}$ \citep[][]{milgrom1982,mcadams2007}. Thus, these bidding functions are linked by the relationship
\[
b_{\theta}^{\rm 1st}({v}) = \int_0^{v} b_\theta^{\rm 2nd}(v') d L_\theta (v' \mid v).
\]
While revenue equivalence fails in this setting, Appendix~\ref{app:correlated} exploits this relationship to extend the comparison in Proposition~\ref{prop:example}:

\begin{prop}\label{prop:correlated} \
In the above correlated value setting, suppose that $(\theta,b^{\rm 2nd}_\theta(v_1),\ldots,b^{\rm 2nd}_\theta(v_n))$ is affiliated. Then $G^{\rm 1st}$ is more accurate than $G^{\rm 2nd}$.
\end{prop}

Proposition~\ref{prop:correlated} imposes affiliation on the second-price auction bids, which is satisfied in natural settings (e.g., the private-value setting where $u(v_i, v_{-i}) =v_i$). We also note that, although bids are correlated conditional on each state, Corollary~\ref{cor:correlated} in Online Appendix~\ref{app:general lehmann} extends the payoff comparison in Corollary~\ref{cor:lehmann} to the current setting.

Second, suppose that the auctioneer has $m > 1$ identical units of the good for sale. The $n > m$ buyers each have unit demand, and their values $v_i$ for one unit are drawn i.i.d.\ from $F_\theta$ conditional on $\theta$, with the same assumptions on $(F_\theta)$ as in Section~\ref{sec:model}. 

We compare the following multi-unit generalizations of second- and first-price auctions, where buyers simultaneously submit a single bid and the $m$ highest bidders receive one unit each. Under the  
\textit{\textbf{VCG (uniform-price) auction}}, the $m$ highest bidders all pay the $(m+1)$th-highest bid (i.e., $\psi_\ell (b_1, \ldots, b_n) = b^{(m+1)}$ for all $\ell \leq m$), and other bidders pay nothing. This reduces to the second-price auction when $m = 1$. The unique symmetric monotone equilibrium remains $b^{\rm vcg}_\theta(v)=v$, with corresponding bid distributions $G^{\rm vcg}=(G_\theta^{\rm vcg})$.
Under the \textit{\textbf{discriminatory auction}}, the $m$ highest bidders all pay their own bids (i.e., $\psi_\ell (b_1, \ldots, b_n) = b^{(\ell)}$ for all $\ell \leq m$), and other bidders pay nothing. This reduces to the first-price auction when $m = 1$. Let $v_{-i}^{(m)}$ denote the $m$th-highest value among the $n-1$ bidders other than $i$. The unique symmetric monotone equilibrium \citep[see Section 13.5.2 in][]{krishna2009} is
\[
b^{\rm disc}_\theta(v)={\mathbb E}_\theta[v_{-i}^{(m)}\mid v_{-i}^{(m)}\leq v],
\]
with bid distributions $G^{\rm disc}=(G_\theta^{\rm disc})$. The MLRP of $(F_\theta)$ implies that $b^{\rm disc}_\theta(v)$ is increasing in $\theta$ for every $v$, whereas $b^{\rm vcg}_\theta(v)=v$ is state-independent. This yields the following generalization of Proposition~\ref{prop:example}:

\begin{prop}\label{prop:multiple} \
$G^{\rm disc}$ is more accurate than $G^{\rm vcg}$. 
\end{prop}

\appendix
\section{Appendix: Proofs and Other Omitted Details}

Throughout, we denote by $b^{-1}_\theta(b)$ (resp.\ $G^{-1}_\theta(x)$) the generalized left (resp.\ right) inverse of an equilibrium bidding function $b_\theta$ (resp.\ bid cdf $G_\theta$), as defined in footnotes~\ref{fn:left-inverse} and \ref{fn:right-inverse}.

\subsection{Proof of Lemma~\ref{lem:competition}}\label{app:competition}

Note that, for every $b\in\mathbb R_+$, 
\begin{align*}
G^{-1}_\theta\!\left(\hat G_\theta(b)\right)
&=\sup\left\{b'\in {\rm supp}\,G_\theta:
F_\theta\!\left(b_\theta^{-1}(b')\right)
\leq F_\theta\!\left(\hat b_\theta^{-1}(b)\right)\right\}\\
&=\sup\left\{b'\in {\rm supp}\,G_\theta:
b_\theta^{-1}(b')\leq \hat b_\theta^{-1}(b)\right\} =b_\theta\!\left(\hat b_\theta^{-1}(b)\right).
\end{align*}
The second equality uses strict monotonicity of $F_\theta$ on $I_\theta$, and the last equality uses continuity and strict monotonicity of $b_\theta$. The claimed equivalence now follows directly from the definition of the accuracy order. \qed

\subsection{Properties of Uniform Order Statistics}\label{app:uniform-order}

Let $x^{(1)}\geq\dots \geq x^{(n)}$ denote the order statistics of $n$ i.i.d.\ draws from the uniform distribution on $[0, 1]$. The probability density function of $x^{(i)}$ is given by
\[
\phi_i(x)=\frac{n!}{(n-i)!(i-1)!}x^{n-i}(1-x)^{i-1}, \quad \forall x \in [0, 1].
\]
For any realized draw $x$, denote by $p_i(x)=\frac{\phi_i(x)}{\sum_j\phi_j(x)}$ 
the probability that $x$ is ranked $i$th-highest. Then, for all $i =1,\ldots, n-1$, 
\[
\frac{p_{i+1}(x)}{p_i(x)}=\frac{\phi_{i+1}(x)}{\phi_i(x)}=\frac{n-i}{i}\frac{1-x}{x},
\]
which is decreasing in $x$.  This also implies that for all $i=1,\ldots, n$,  $\frac{p_i(x)}{p_1(x)}$ is decreasing and $\frac{p_i(x)}{p_n(x)}$ is increasing in $x$. More generally, we have the following result:

\begin{lem}\label{lem:alpha} Take non-zero vectors $\alpha,\alpha'\in\mathbb R_+^n$ such that $\alpha_j'\alpha_i\geq\alpha_i'\alpha_j$ for all $i,j$ with $i<j$. Then $\frac{\sum_i\alpha'_ip_i(x)}{\sum_i\alpha_ip_i(x)}$ is decreasing in $x\in (0,1)$.
\end{lem}

\begin{proof}
For all $x, x' \in (0, 1)$ with $x<x'$, we have
\begin{eqnarray*}
&&\frac{\sum_i\alpha_i'p_i(x)}{\sum_i\alpha_ip_i(x)}\geq\frac{\sum_i\alpha_i'p_i(x')}{\sum_i\alpha_ip_i(x')}\\
&\iff&\left(\sum_i\alpha_i'p_i(x) \right)\left(\sum_i\alpha_ip_i(x')\right)\geq \left(\sum_i\alpha_i'p_i(x') \right) \left(\sum_i\alpha_ip_i(x) \right)\\
&\iff& \sum_i \sum_{j\not=i} \alpha_i'p_i(x)\alpha_jp_j(x')\geq \sum_i \sum_{j\not=i} \alpha'_ip_i(x')\alpha_jp_j(x)\\
&\iff& \sum_i \sum_{j\not=i} \alpha_i'\alpha_j(p_i(x)p_j(x')-p_i(x')p_j(x))\geq 0\\
&\iff& \sum_i \sum_{j>i} \alpha_i'\alpha_j(p_i(x)p_j(x')-p_i(x')p_j(x)) +  \alpha_j'\alpha_i(p_j(x)p_i(x')-p_j(x')p_i(x))\geq 0\\
&\iff& \sum_i \sum_{j>i} (\alpha_j'\alpha_i-\alpha_i'\alpha_j)p_i(x)p_i(x')\left(\frac{p_j(x)}{p_i(x)}-\frac{p_j(x')}{p_i(x')}\right)\geq 0.
\end{eqnarray*}
The final inequality holds by the assumption that $\alpha_j'\alpha_i\geq\alpha_i'\alpha_j$ for all $i,j$ with $i<j$.
\end{proof}

\subsection{Proof of Theorem~\ref{thm:main}}\label{app:main}

Denote by $x^{(1)}\geq x^{(2)}\geq \ldots \geq x^{(n)}$ the order statistics corresponding to the $n$ value quantiles $x_i = F_\theta (v_i)$, which in each state $\theta$ are i.i.d.\ uniform on $[0, 1]$. As in Appendix~\ref{app:uniform-order}, let 
$p_i(x) = \frac{\phi_i(x)}{\sum_j\phi_j(x)}$ denote the probability that a realized quantile $x$ is $i$th-highest. 

Let $b^{\rm 1st}_\theta$ and $b^{\psi}_\theta$ denote the equilibrium bidding functions under the first-price and $\psi$-auction as  functions of value quantiles $x$. Since $(G^\psi_\theta)$ is FOSD-increasing in $\theta$ (by the MLRP of $(G^\psi_\theta)$), $b^{\psi}_\theta (x)$ is increasing in $\theta$ for all $x$. To prove that $G^{\rm 1st}$ is more accurate than $G^{\psi}$, it suffices, by Lemma~\ref{lem:competition}, to show that $b^{\rm 1st}_\theta((b_\theta^\psi)^{-1}(b))\geq b^{\rm 1st}_{\theta'}((b_{\theta'}^\psi)^{-1}(b))$ for all $b\in\mathbb R_+$ and $\theta>\theta'$.

We first derive a tractable expression for $b^{\rm 1st}_\theta((b_\theta^\psi)^{-1}(b))$ that holds whenever $b \in (b_\theta^\psi (0), b_\theta^\psi (1))$. Since $\psi_\ell (b^{(1)}, \ldots, b^{(n)})=0$ whenever $b^{(\ell)}=0$, every quantile $x$ receives a non-negative interim payoff in equilibrium. Thus, since $p_n(0)=1$,  the interim equilibrium payoff of $x=0$ under $\psi$ is 0. Hence, by revenue equivalence, the expected payment of each quantile $x$ is the same across the first-price and $\psi$-auction. That is, for each quantile $x\in (0,1)$, 
\begin{equation}\label{eq:rev-equivalence}
p_{1}(x) b^{\rm 1st}_\theta(x)=  \sum_{i=1}^n p_{i}(x){\mathbb E}[ \psi_i(b_\theta^\psi(x^{(1)}),\ldots, b_\theta^\psi(x^{(n)})) \mid x^{(i)}=x ]. 
\end{equation}
Thus, for any $b \in (b_\theta^\psi (0), b_\theta^\psi (1))$, substituting $x=(b_\theta^\psi)^{-1}(b)$ into (\ref{eq:rev-equivalence}) and dividing by $p_1(x) > 0$ yields
\begin{equation}\label{eq:main}
b^{\rm 1st}_\theta((b_\theta^\psi)^{-1}(b))=  \sum_{i=1}^n \frac{p_{i}((b_\theta^\psi)^{-1}(b))}{p_1((b_\theta^\psi)^{-1}(b))} {\mathbb E}[ \psi_i(b_\theta^\psi(x^{(1)}),\ldots, b_\theta^\psi(x^{(n)})) \mid x^{(i)}=(b_\theta^\psi)^{-1}(b)].
\end{equation}
 
Next, we fix any $\theta > \theta'$ and prove that $b^{\rm 1st}_\theta((b_\theta^\psi)^{-1}(b))\geq b^{\rm 1st}_{\theta'}((b_{\theta'}^\psi)^{-1}(b))$ for all $b \in \mathbb{R}_+$. 
We first assume that $b \in(b_\theta^\psi(0),b_{\theta'}^\psi(1))$, which in particular implies that $b \in (b_\theta^\psi(0),b_{\theta}^\psi(1))$ and $b \in (b_{\theta'}^\psi(0),b_{\theta'}^\psi(1))$. Thus, (\ref{eq:main}) applies to both $b^{\rm 1st}_\theta((b_\theta^\psi)^{-1}(b))$ and $b^{\rm 1st}_{\theta'} ((b_{\theta'}^\psi)^{-1}(b))$, so it suffices to establish the following two claims:

\medskip

\noindent{\bf Claim 1:} For each $i$, $\frac{p_{i}((b_\theta^\psi)^{-1}(b))}{p_1((b_\theta^\psi)^{-1}(b))} \geq \frac{p_{i}((b_{\theta'}^\psi)^{-1}(b))}{p_1((b_{\theta'}^\psi)^{-1}(b))}$.


\noindent\textit{Proof of Claim 1.} This is immediate from the fact that $(b_\theta^\psi)^{-1}(b) \leq (b_{\theta'}^\psi)^{-1}(b)$, and that $\frac{p_{i}(x)}{p_1(x)}$ is decreasing in $x$ by Lemma~\ref{lem:alpha}. \qed

\medskip
\noindent{\bf Claim 2:} For each $i$, we have
\[{\mathbb E}[ \psi_i(b_\theta^\psi(x^{(1)}),\ldots, b_\theta^\psi(x^{(n)})) \mid x^{(i)}=(b_\theta^\psi)^{-1}(b) ] \geq {\mathbb E}[ \psi_i(b_{\theta'}^\psi(x^{(1)}),\ldots, b_{\theta'}^\psi(x^{(n)})) \mid x^{(i)}=(b_{\theta'}^\psi)^{-1}(b) ].\]

\noindent\textit{Proof of Claim 2.} For each $\tilde \theta \in \{ \theta, \theta' \}$, the fact that $b \in (b_{\tilde \theta}^\psi(0),b_{\tilde \theta}^\psi(1))$ implies that
\[
{\mathbb E}[ \psi_i(b_{\tilde \theta}^\psi(x^{(1)}),\ldots, b_{\tilde \theta}^\psi(x^{(n)})) \mid x^{(i)}=(b_{\tilde \theta}^\psi)^{-1}(b) ]={\mathbb E}_{\tilde \theta} [ \psi_i(b^{(1)},\ldots, b^{(n)}) \mid b^{(i)}=b],
\]
where $b^{(1)},\ldots, b^{(n)}$ denote the order statistics of $n$ i.i.d.\ draws from distribution $G_{\tilde \theta}^{\psi}$. Moreover, we can analyze ${\mathbb E}_{\tilde \theta} [ \psi_i(b^{(1)},\ldots, b^{(n)}) \mid b^{(i)}=b]$ using the Markov property of order statistics \citep[Chapter 5 in][]{ahsanullah2013}. Specifically, denote by $G_{\tilde \theta}^{\psi-}$ and $G_{\tilde \theta}^{\psi+}$ the conditional distributions of $G_{\tilde \theta}^\psi$  on $[0, b]$ and $[b, \infty)$, respectively. Then the distribution of $(b^{(1)},\ldots, b^{(i-1)})$ (resp.\ $(b^{(i+1)},\ldots, b^{(n)})$) conditional on the realization of $b^{(i)}=b$ coincides with the distribution of order statistics of $i-1$ i.i.d.\ draws from  $G_{\tilde \theta}^{\psi+}$ (resp.\ $n-i$ i.i.d.\ draws from $G_{\tilde \theta}^{\psi-}$). Moreover, $(b^{(1)},\ldots, b^{(i-1)})$  and $(b^{(i+1)},\ldots, b^{(n)})$  are independently distributed conditional on $b^{(i)}=b$.  

Since both $(G_{\tilde \theta}^{\psi-})$ and $(G_{\tilde \theta}^{\psi+})$ are FOSD-increasing in $\tilde \theta$ by the MLRP of $(G^{\psi}_{\tilde \theta})$ \citep[][Theorem 1.C.6]{shaked2007} and since $\psi_i$ is increasing, this implies that
\[
{\mathbb E}_\theta[ \psi_i(b^{(1)},\ldots, b^{(n)}) \mid b^{(i)}=b]\geq {\mathbb E}_{\theta'}[ \psi_i(b^{(1)},\ldots, b^{(n)}) \mid b^{(i)}=b],
\]
proving Claim 2. \qed

This establishes $b^{\rm 1st}_\theta((b_\theta^\psi)^{-1}(b))\geq b^{\rm 1st}_{\theta'}((b_{\theta'}^\psi)^{-1}(b))$ for all $b \in(b_\theta^\psi(0),b_{\theta'}^\psi(1))$. It remains to show that this inequality in fact holds for all $b \in \mathbb{R}_+$. 
To show this, first suppose that $b_\theta^\psi(0)<b_{\theta'}^\psi(1)$. Then, by continuity, $b^{\rm 1st}_\theta((b_\theta^\psi)^{-1}(b))\geq b^{\rm 1st}_{\theta'}((b_{\theta'}^\psi)^{-1}(b))$ also holds at $b=b_\theta^\psi(0)$.  Thus, for $b\leq b_\theta^\psi(0)$,
\begin{eqnarray*}
b_\theta^{\rm 1st}((b_\theta^\psi)^{-1}(b))
=b_\theta^{\rm 1st}(0)&=&\lim_{b'\searrow b^\psi_\theta(0)}b_{\theta}^{\rm 1st}((b_{\theta}^\psi)^{-1}(b'))\\
&\geq& \lim_{b'\searrow b^\psi_\theta(0)}b_{\theta'}^{\rm 1st}((b_{\theta'}^\psi)^{-1}(b'))=b_{\theta'}^{\rm 1st}
   ((b_{\theta'}^\psi)^{-1}(b_\theta^\psi(0)))
\geq b_{\theta'}^{\rm 1st}((b_{\theta'}^\psi)^{-1}(b)).
\end{eqnarray*}
An analogous argument accommodates any $b\geq b_{\theta'}^\psi(1)$. 
Next suppose that $b_\theta^\psi(0)\geq b_{\theta'}^\psi(1)$. Using the revenue-equivalence identity
(\ref{eq:rev-equivalence}), one can show $b_\theta^{\rm 1st}(\varepsilon)
\geq b_{\theta'}^{\rm 1st}(1-\varepsilon)$ for all sufficiently small $\varepsilon>0$, which then ensures  $b_\theta^{\rm 1st}(0)\geq b_{\theta'}^{\rm 1st}(1)$.
Consequently, for all $b \in \mathbb{R}_+$,
\[
b_\theta^{\rm 1st}((b_\theta^\psi)^{-1}(b))
\geq b_\theta^{\rm 1st}(0)
\geq b_{\theta'}^{\rm 1st}(1)
\geq b_{\theta'}^{\rm 1st}((b_{\theta'}^\psi)^{-1}(b)),
\]
as required. \qed

\subsection{Proof of Proposition~\ref{prop:k}}\label{app:k}

We prove the following more general claim: For any $k=1,\ldots, n-1$ and $\psi$ given by $\psi_1(b^{(1)},\ldots, b^{(n)})=\sum_{i\geq k}\beta_i b^{(i)}$ for some non-zero vector $\beta\in\mathbb R_+^{n-k+1}$ and $\psi_\ell \equiv 0$ for all $\ell \geq 2$, $G^{k{\rm th}}$ is more accurate than $G^\psi$ whenever the latter satisfies the MLRP.

We first show that for each $\theta$, the equilibrium bidding function under the $k$th-price auction (viewed as a function of quantiles) can be written as 
\begin{equation}\label{eq:k recursion}
b^{k{\rm th}}_\theta(x)={\mathbb E} \left[ \sum_{i\geq k}\beta_ib^\psi_\theta(x^{(i)}) \mid x^{(k)}=x \right].
\end{equation}
To see this, note that
\begin{eqnarray*}
b^{{\rm 1st}}_\theta(x)&=&{\mathbb E} \left[\sum_{i\geq k}\beta_ib^\psi_\theta(x^{(i)})  \mid x^{(1)}=x \right] 
= {\mathbb E}\left[{\mathbb E}\left[\sum_{i\geq k}\beta_ib^\psi_\theta(x^{(i)}) \mid x^{(1)}=x, x^{(k)} \right]  \mid x^{(1)}=x \right] \\
&=& {\mathbb E}\left[{\mathbb E} \left[\sum_{i\geq k}\beta_ib^\psi_\theta(x^{(i)}) \mid x^{(k)} \right] \mid x^{(1)}=x \right],
\end{eqnarray*}
where the first equality uses (\ref{eq:rev-equivalence}), the second equality uses the law of iterated expectations, and the third equality uses the Markov property of order statistics.  Thus, the function 
 $b_\theta(x):={\mathbb E}\left[ \sum_{i\geq k}\beta_ib^\psi_\theta(x^{(i)}) \mid x^{(k)}=x \right]$ satisfies $
b^{{\rm 1st}}_\theta(x)= {\mathbb E} \left[ b_\theta(x^{(k)}) \mid x^{(1)}=x \right]$.

This implies that the expected equilibrium payment of quantile $x$ under the first-price auction is the same as the expected payment of quantile $x$ under the $k$th-price auction if bidders follow strategy $b_\theta$. Note also that $b_\theta(x)$ is strictly increasing and continuous in $x$, because the distribution of $(x^{(k+1)},\ldots, x^{(n)})$ conditional on $x^{(k)}=x$ coincides with the  distribution of order statistics of $n-k$ i.i.d.\ draws from the uniform distribution on $[0, x]$. Hence, revenue-equivalence arguments imply that $b_\theta (x)$ is a (symmetric monotone) equilibrium of the $k$th-price auction. 
  Since the equilibrium of the $k$th-price auction is unique \citep[see][Theorem A]{monderer}, this yields (\ref{eq:k recursion}).

For each bid $b$, (\ref{eq:k recursion}) implies
\[
b^{k{\rm th}}_\theta((b_\theta^\psi)^{-1}(b))=  {\mathbb E} \left[\sum_{i\geq k}\beta_ib^\psi_\theta(x^{(i)}) \mid x^{(k)}=(b_\theta^\psi)^{-1}(b) \right]. 
\] 
The same argument as in the proof of Theorem~\ref{thm:main} shows that the right-hand side is increasing in $\theta$, using the MLRP of $G^\psi$. Hence, $G^{k{\rm th}}$ is more accurate than $G^\psi$ by Lemma~\ref{lem:competition}. 
\qed

\subsection{Proof of Proposition~\ref{prop:own bid}}\label{app:own bid}

Note that $\sum_{i=1}^n\alpha_i p_i(x)$ is strictly positive for every $x\in(0,1)$. By the revenue-equivalence argument \citep[e.g., Chapter 3 in ][]{krishna2009},  there is a unique equilibrium, which satisfies 
\begin{equation}\label{eq:alpha bid}
b^\alpha_\theta(x)=b^{\rm 1st}_\theta(x)
\frac{p_1(x)}{\sum_{i=1}^n\alpha_i p_i(x)}, \qquad x\in(0,1).
\end{equation}
For the loser-pay auction, (\ref{eq:alpha bid}) becomes $b^{\rm loser}_\theta(x)=b^{\rm 1st}_\theta(x)\left(\frac{x}{1-x}\right)^{n-1}$.

Applying (\ref{eq:alpha bid}) to both $b^{\alpha}_\theta$ and $b^{\alpha'}_\theta$ yields that for all $x\in(0,1)$,
\[
b^\alpha_\theta(x)=\frac{\sum_i\alpha'_ip_i(x)}{\sum_i\alpha_ip_i(x)}b^{\alpha'}_\theta(x).
\]
For $b$ whose induced quantiles are interior, substituting $x=(b_\theta^{\alpha'})^{-1}(b)$ gives
\[
b^{\alpha}_\theta\left( \left(b_\theta^{\alpha'}\right)^{-1}(b) \right)=  \frac{ \sum_{i}  \alpha'_i p_{i} \left( \left(b_\theta^{\alpha'} \right)^{-1}(b)\right)}{ \sum_{i} \alpha_i p_{i}\left( \left(b_\theta^{\alpha'}\right)^{-1}(b) \right)} \, b_\theta^{\alpha'} \left( \left(b_\theta^{\alpha'}\right)^{-1}(b) \right).
\]
The ratio on the right is increasing in $\theta$, because $(b_\theta^{\alpha'})^{-1}(b)$ is decreasing in $\theta$ and $\frac{\sum_i\alpha'_ip_i(x)}{\sum_i\alpha_ip_i(x)}$ is decreasing in $x$ by Lemma~\ref{lem:alpha}. The last factor is also increasing in $\theta$, by the generalized-inverse convention and the fact that $b^{\alpha'}_\theta$ is increasing in $\theta$. The case of boundary quantiles  follows by a continuity argument analogous to the proof of Theorem~\ref{thm:main}. Hence, $b^{\alpha}_\theta((b_\theta^{\alpha'})^{-1}(b))$ is increasing in $\theta$, which implies that $G^{\alpha}$ is more accurate than $G^{\alpha'}$ by Lemma~\ref{lem:competition}. \qed

\subsection{Proof of Proposition~\ref{prop:WOA}}

The equilibrium bidding function under the war of attrition is 
\[
b^{\mathrm{woa}}_\theta(v)
=\int_0^v w\frac{h_\theta(w)}{1-H_\theta(w)}\,dw,
\]
where $H_\theta(v)=F_\theta(v)^{n-1}$ and $h_\theta$ denotes its density. 
Since $h_\theta(w)/(1-H_\theta(w))$ is decreasing in $\theta$ by the MLRP of $(F_\theta)$, it follows that
$b^{\mathrm{woa}}_\theta(v)$ is decreasing in $\theta$ for every fixed $v$. Thus,
$G^{\mathrm{woa}}$ is less accurate than $G^{2\mathrm{nd}}$ by Lemma~\ref{lem:competition}.

We write the all-pay and war-of-attrition bids as functions of
quantiles:
\[
b^{\mathrm{all}}_\theta(x)=\int_0^x F_\theta^{-1}(y)(n-1)y^{n-2}\,dy, \;\;  \quad b^{\mathrm{woa}}_\theta(x)=\int_0^x F^{-1}_\theta(y)
\frac{(n-1)y^{n-2}}{1-y^{n-1}}\,dy.
\]

Since $F_\theta$ is FOSD-increasing in $\theta$, $b^{\mathrm{woa}}_\theta(x)$ is increasing in $\theta$. 
Hence, for each bid $b$, the corresponding quantile
$
x_\theta(b)=(b^{\mathrm{woa}}_\theta)^{-1}(b)
$
is decreasing in $\theta$. For each $b>0$,
\begin{align*}
\frac{d}{db}b^{\mathrm{all}}_\theta\!\left(x_\theta(b)\right)
&=\frac{(b^{\mathrm{all}}_\theta)'(x_\theta(b))}{(b^{\mathrm{woa}}_\theta)'(x_\theta(b))} =1-x_\theta(b)^{n-1}.
\end{align*}
Since $x_\theta(b)$ is decreasing in $\theta$, the derivative above is increasing in
$\theta$. Moreover, $b^{\mathrm{all}}_\theta\!\left(x_\theta(0)\right)=0$ for every $\theta$. Therefore,
\[
b^{\mathrm{all}}_\theta\!\left(x_\theta(b)\right)
=\int_0^b\left(1-x_\theta(s)^{n-1}\right)ds
\]
is increasing in $\theta$ for every $b$, which implies $G^{\mathrm{all}}$ is more accurate than $G^{\mathrm{woa}}$  by Lemma~\ref{lem:competition}.
\qed

\subsection{Proof of Corollary~\ref{cor:prediction}}\label{app:cor-prediction}

We assume without loss that $\Theta=[0,1]$, and let $q$ denote the prior density on $\Theta$. 
We construct value distributions $(F^{K, \eta}_\theta)$ indexed by $K, \eta>0$, whose densities at each $v\in [0,1]$ are
\[
f_\theta^{K,\eta}(v)
=
\frac{e^{K\theta v}h_\eta(v)}
{\int_0^1 e^{K\theta w}h_\eta(w)\,dw}, \quad \quad \text{ where } 
h_\eta(v)
=
\frac{e^{-v/\eta}+e^{-(1-v)/\eta}}
     {2\eta(1-e^{-1/\eta})}.
\]
The MLRP holds because, for $\theta'>\theta$,
$
\frac{f_{\theta'}^{K,\eta}(v)}{f_\theta^{K,\eta}(v)}
=C_{\theta,\theta'}e^{K(\theta'-\theta)v},
$ for some constant $C_{\theta,\theta'}$. Note that $h_\eta$ is the density function of the mixture distribution $\tfrac12(H_{0,\eta}+H_{1,\eta})$, where $H_{0,\eta}$ and $H_{1,\eta}$ have respective densities
$
\frac{e^{-v/\eta}}{\eta(1-e^{-1/\eta})}
$ and $
\frac{e^{-(1-v)/\eta}}{\eta(1-e^{-1/\eta})}$.

Below, we consider the expected minimal losses $L(G_{m}^{{\rm 1st}, K, \eta})$ and $L(G_{n}^{{\rm 2nd}, K, \eta})$ under the first- and second-price auction with $m$ and $n$ buyers, respectively.


\smallskip

\noindent{\bf Second-price auction:}
Let $(\hat G^{{\rm 2nd}, K}_\theta)$ be an auxiliary profile of distributions with binary support $\{0,1 \}$ under which  the probability of observing 0 in state $\theta$ is $p_K(\theta)=\frac{1}{1+e^{K\theta}}$.
Given $n$ i.i.d.\ draws $(b_1, \ldots, b_n)$ from $(\hat G^{{\rm 2nd}, K}_\theta)$, consider drawing $b'_i\sim H_{b_i,\eta}$, independently across $i$. Thus,  conditional on $\theta$, each $b'_i$ is drawn i.i.d. from $p_K(\theta)H_{0,\eta}
+(1-p_K(\theta))H_{1,\eta}$. 

Let $\sigma_\eta$ be an optimal strategy based on the $n$ second-price auction bids drawn from $(G_\theta^{{\rm 2nd}, K,\eta})^{\otimes n}$. Suppose the DM instead observes $n$ i.i.d.\ draws from $(\hat G^{{\rm 2nd}, K}_\theta)^{\otimes n}$, and uses these to draw $(b'_1,\ldots, b'_n)$ as described above and then chooses $\sigma_\eta(b_1',\ldots, b'_n)$.  Observe 
\begin{eqnarray*}
&&\left\|
(G_\theta^{{\rm 2nd}, K,\eta})^{\otimes n}
-(p_K(\theta)H_{0,\eta}
+(1-p_K(\theta))H_{1,\eta})^{\otimes n}
\right\|_{\mathrm{TV}}  \\
&\leq& n \sup_{\theta\in[0,1]}
\left\|
F_\theta^{K,\eta}
-\left[p_K(\theta)H_{0,\eta}
+(1-p_K(\theta))H_{1,\eta}\right]
\right\|_{\mathrm{TV}}.
\end{eqnarray*}
The right-hand side goes to zero as $\eta\to 0$. Since losses are bounded by $\ell(1)$, by using the strategy described above under $(\hat G^{{\rm 2nd}, K}_\theta)^{\otimes n}$,  the DM can secure 
\[
L(\hat G^{{\rm 2nd}, K}_n)\leq \liminf_{\eta\to 0}L(G_{n}^{{\rm 2nd}, K, \eta}).  
\]

Under $(\hat G^{{\rm 2nd}, K}_\theta)^{\otimes n}$, the prior probability that at least one zero draw is observed is
at most $n\int_0^1 p_K(\theta)q(\theta)\,d\theta$, which goes to 0 as $K\to\infty$ by dominated convergence. 
Hence, for all large enough $K$, $
L(\hat G^{{\rm 2nd}, K}_n)\geq L_\emptyset- \varepsilon/2
$.
Given any such $K$, we can choose $\eta$ small enough that $L(G_{n}^{{\rm 2nd}, K, \eta})\geq L_\emptyset-\varepsilon$.

\smallskip

\noindent{\bf First-price auction:} The first-price bid of value quantile $x$ is
\[
(G_{\theta}^{{\rm 1st}, K,\eta})^{-1}(x)
=
\frac{m-1}{x^{m-1}}
\int_0^x (F_\theta^{K,\eta})^{-1}(t) t^{m-2}\,dt.
\]
Note that,  as $\eta\to 0$,
$
(F_\theta^{K,\eta})^{-1}(x)\to 
\mathbf{1}\{x>p_K(\theta)\}
$ for every $\theta$ and almost every $x$. 
Thus, dominated convergence gives
\[
(G_{\theta}^{{\rm 1st}, K,\eta})^{-1}(x)
\longrightarrow
(G_{\theta}^{{\rm 1st}, K, 0})^{-1}(x):
=
\begin{cases}
0
& \text{ if } x\leq p_K(\theta),\\[1mm]
1-\left(\dfrac{p_K(\theta)}{x}\right)^{m-1} 
& \text{ if } x>p_K(\theta).
\end{cases}
\]

Consider a strategy $\sigma$ that only uses the first bid $b_1$ among the observations $b=(b_1,\ldots, b_m)$:
\[
\sigma(b)=\max\left\{0, \min\left\{1, \left(
-\frac{\log(1-b_1)}{(m-1)K}
\right)\right\} \right\}.
\]
For $x>p_K(\theta)$,
\[
\left|
\sigma
\bigl((G_{\theta}^{{\rm 1st}, K, 0})^{-1}(x),b_2, \ldots, b_m \bigr)-\theta
\right|
\leq
\frac{\log2+|\log x|}{K}.
\]
Under $\eta=0$, the probability of $G_{\theta}^{{\rm 1st}, K, \eta}(b_1)\leq p_K(\theta)$ at $\theta$ is $p_K(\theta)\leq \frac{1}{2+\theta K}$. 
Hence,
\[
\limsup_{\eta\to 0}
L(G_{m}^{\mathrm{1st}, K, \eta})
\leq  \ell(1)\int \frac{1}{2+\theta K} q(\theta)d\theta 
+\int_0^1
\ell\left(
\min\left\{1,\frac{\log2+|\log x|}{K}\right\}
\right)dx.
\]
The right-hand side goes to 0 as $K\to\infty$ by dominated convergence. Pick $K$ large enough that the right-hand side is strictly less than $\varepsilon$. 
 Then for all small enough $\eta$, we have $L(G_{m}^{\mathrm{1st}, K, \eta})\leq \varepsilon$. \qed

\subsection{Proof of Proposition~\ref{prop:reserve}}

\noindent{\bf First part:} Writing bids as functions of quantiles,  for every $b\in\mathbb R_+$, we obtain 
$
(G^{{\rm 1st}, r}_\theta)^{-1}\!\left(G^{\psi,r}_\theta(b)\right)
=b_\theta^{{\rm 1st}, r}\!\left(G^{\psi,r}_\theta(b)\right)
$
as in the proof of Lemma~\ref{lem:competition}. For $0\leq b\leq r$, we have 
\[
b_\theta^{{\rm 1st}, r}\!\left(G^{\psi,r}_\theta(b)\right)=b_\theta^{{\rm 1st}, r}\!\left(F_\theta(r)\right)=
\begin{cases}
0 \text{ for } r \geq \overline v_\theta \\
r \text{ for } r\in [\underline v_\theta, \overline v_\theta) \\
\underline v_\theta \text{ for } r<\underline v_\theta, 
\end{cases}
\]
which is increasing in $\theta$ by the MLRP of $(F_\theta)$.

For $b>r$, $G^{\psi,r}_\theta(b)=(b^{\psi, r}_\theta)^{-1}(b)\geq F_\theta(r)$ for each $\theta$. Observe that revenue-equivalence arguments give the analog of (\ref{eq:rev-equivalence}) for every $x\in [F_\theta(r), 1)$. Based on this, the  same  argument as in the proof of Theorem~\ref{thm:main} shows that
$
b^{{\rm 1st},r}_\theta\!\left((b^{\psi, r}_\theta)^{-1}(b)\right)
$  is increasing in $\theta$.

\smallskip 
\noindent{\bf Second part:} Suppose $0\leq r'<r$. Applying integration by parts to the closed-form bidding function $b^{\rm 1st}_\theta$ yields that for any $x\geq F_\theta(r)$ with $x\in (0,1)$, 
\begin{equation}\label{eq:reserve-difference}
b^{{\rm 1st}, r'}_{\theta}(x)
=b^{{\rm 1st}, r}_{\theta}(x)-\int_{r'}^r\left(\frac{F_\theta(t)}{x}\right)^{n-1}dt.
\end{equation}
Moreover, for every bid $b$,
\[
(G^{{\rm 1st},r'}_\theta)^{-1}\!\left(G^{{\rm 1st},r}_\theta(b)\right)
=b^{{\rm 1st},r'}_\theta\!\left(G^{{\rm 1st},r}_\theta(b)\right),
\]
where bids at quantiles $0$ and $1$ are interpreted as the corresponding endpoint limits. 

Fix $\theta>\theta'$. We need to show that for every $b$, $b^{{\rm 1st},r'}_\theta\!\left(G^{{\rm 1st},r}_\theta(b)\right)
\geq
b^{{\rm 1st},r'}_{\theta'} \!\left(G^{{\rm 1st},r}_{\theta'} (b)\right)$. We consider two cases:

\noindent{\bf Case (i):} At both states $\theta$ and $\theta'$,  there are participating types when the reserve price is $r$.

For a fixed bid $b$, 
we focus on the case $G^{{\rm 1st},r}_\theta(b),G^{{\rm 1st},r}_{\theta'}(b)>0$; the zero-endpoint cases then follow by right limits. The endpoints of the positive-bid support at reserve $r$ are increasing in $\theta$: Indeed, the lower endpoint is $\max\{r,\underline v_\theta\}$, while the upper endpoint is $\int \max\{w, r\}  dF_\theta^{n-1}(w)$, which is increasing by the MLRP of $(F_\theta)$. It follows that
\begin{equation}\label{eq:reserve-projected-bids}
b^{{\rm 1st},r}_\theta(G^{{\rm 1st},r}_\theta(b))
\geq b^{{\rm 1st},r}_{\theta'}(G^{{\rm 1st},r}_{\theta'}(b)),
\end{equation}
which can be written as
\[
\int_r^\infty\left[
\min\left\{\left(\frac{F_\theta(t)}{G^{{\rm 1st},r}_{\theta}(b)}\right)^{n-1},1\right\}
-\min\left\{\left(\frac{F_{\theta'}(t)}{G^{{\rm 1st},r}_{\theta'}(b)}\right)^{n-1},1\right\}
\right]dt\leq0.
\]
Observe that $\left(\frac{F_\theta(t)}{F_{\theta'}(t)}\frac{G^{{\rm 1st},r}_{\theta'}(b)}{G^{{\rm 1st},r}_\theta(b)}\right)^{n-1}$ is increasing in $t$ by the MLRP of $(F_\theta)$. Thus, the integrand $
\min\left\{\left(\frac{F_\theta(t)}{G^{{\rm 1st},r}_{\theta}(b)}\right)^{n-1},1\right\}
-\min\left\{\left(\frac{F_{\theta'}(t)}{G^{{\rm 1st},r}_{\theta'}(b)}\right)^{n-1},1\right\}$ crosses zero at most once, from negative to positive.  Hence, the above inequality implies 
$
\left(\frac{F_\theta(t)}{G^{{\rm 1st},r}_\theta(b)}\right)^{n-1}
\leq
\left(\frac{F_{\theta'}(t)}{G^{{\rm 1st},r}_{\theta'}(b)}\right)^{n-1}$
for every $t\leq r$.

Using (\ref{eq:reserve-difference}) and (\ref{eq:reserve-projected-bids}), we conclude that
\begin{align*}
b^{{\rm 1st},r'}_\theta(G^{{\rm 1st},r}_\theta(b))-b^{{\rm 1st},r'}_{\theta'}(G^{{\rm 1st},r}_{\theta'}(b))
&=b^{{\rm 1st},r}_\theta(G^{{\rm 1st},r}_{\theta}(b))-b^{{\rm 1st},r}_{\theta'}(G^{{\rm 1st},r}_{\theta'}(b))\\
&\quad+\int_{r'}^r\left[
\left(\frac{F_{\theta'}(t)}{G^{{\rm 1st},r}_{\theta'}(b)}\right)^{n-1}
-\left(\frac{F_\theta(t)}{G^{{\rm 1st},r}_\theta(b)}\right)^{n-1}
\right]dt\geq0,
\end{align*}
as required.

\noindent{\bf Case (ii):}  State $\theta'$ has no participating type when the reserve price is $r$. 

If neither state has a participating type at $r$, both $G^{{\rm 1st},r}_\theta$ and $G^{{\rm 1st},r}_{\theta'}$ are degenerate at zero, and the desired inequality follows because the upper endpoint of bid distributions under $r'$ is increasing in the state. Suppose instead that state $\theta$ has participating types at $r$, and take bid $b$ such that $x=G^{{\rm 1st},r}_\theta(b)\geq F_\theta(r)$ and $G^{{\rm 1st},r}_\theta(b)>0$. If $F_\theta(r)>0$ and $x>0$, then, for every $t\leq r$, the MLRP of $(F_\theta)$ and $F_{\theta'}(r)=1$ give
\[
\frac{F_\theta(t)}{G^{{\rm 1st},r}_\theta(b)}
\leq\frac{F_\theta(t)}{F_\theta(r)}
\leq F_{\theta'}(t).
\]
 If $F_\theta(r)=0$, all values in state $\theta$ are at least $r$, while $F_{\theta'}(r)=1$ by assumption.
Thus, in either case,  $\min\{\left(\frac{F_\theta(\cdot)}{G^{{\rm 1st},r}_\theta(b)}\right)^{n-1},1\}$ first-order stochastically dominates $F^{n-1}_{\theta'}(\cdot)$. Therefore,
\[
b^{{\rm 1st},r'}_\theta\!\left(G^{{\rm 1st},r}_\theta(b)\right)
\geq
b^{{\rm 1st},r'}_{\theta'}(1) \geq b^{{\rm 1st},r'}_{\theta'} \!\left(G^{{\rm 1st},r}_{\theta'} (b)\right),
\]
as required. For any $b$ with $G^{{\rm 1st},r}_\theta(b)=0$, the same conclusion is obtained by a right limit. \qed

\subsection{Proof of Proposition~\ref{prop:interdependent}}\label{app:interdependent}

To prove the first part, we exploit the fact that revenue equivalence continues to apply in our interdependent value setting, as types $v_i$ are assumed i.i.d.\ conditional on each $\theta$ \citep[see, e.g.,][Chapter 7]{krishna2009}. Then the same arguments as in the proof of Theorem~\ref{thm:main} yield (\ref{eq:main}), which can again be used to show that $G^{\rm 1st}$ is more accurate than $G^{\psi}$, provided that $G^{\psi}$ satisfies the MLRP.

For the second part, note that  $b_\theta^{{\rm 2nd}}$ is increasing in $\theta$ while $b_\theta^{\rm vcg}$ is constant in $\theta$. Hence, $G^{\rm 2nd}$ is more accurate than $G^{\rm vcg}$  by Lemma~\ref{lem:competition}. \qed

\subsection{Proof of Proposition~\ref{prop:imperfect} }\label{app:imperfect}

We use the following extension of Lemma~\ref{lem:competition} to the current setting (see Appendix~\ref{sec:gencompproof} for the proof).  Note that, unlike Lemma~\ref{lem:competition} (where $\xi = \theta$), Lemma~\ref{lem:gencompetition} imposes an affiliation condition in order to deal with imperfect state observations.

\begin{lem}\label{lem:gencompetition}
Consider any two auctions with equilibrium bidding functions $(b_\xi)$ and $(\hat b_\xi)$ and bid distributions $G= (G_\theta)$ and $\hat{G}= (\hat G_\theta)$.  Suppose that 
\begin{enumerate}
\item $b_\xi (\hat{b}_\xi^{-1}(b))$ is increasing in $\xi$ for all $b$; and
\item $(\theta, \xi, b_\xi(v))$ are affiliated; and
\item for every $\theta$, $G_\theta$ is continuous.
\end{enumerate}
Then $G$ is more accurate than $\hat{G}$.
\end{lem}

Given Lemma~\ref{lem:gencompetition}, to prove Proposition~\ref{prop:imperfect}, observe that $G^{{\rm 1st}}_\theta$ is continuous for each $\theta$ given the stated assumptions. 
Thus, it suffices to verify that $b^{{\rm 1st}}_\xi (({b}^\psi_\xi)^{-1}(b))$ is increasing in $\xi$. 
Since $(v_1,\ldots, v_n)$ are i.i.d.\ conditional on each $\xi$, revenue equivalence can again be used to derive the following analog of (\ref{eq:main}):
\[
b^{\rm 1st}_\xi((b_\xi^\psi)^{-1}(b))=  \sum_{i} \frac{p_{i}((b_\xi^\psi)^{-1}(b))}{p_1((b_\xi^\psi)^{-1}(b))} {\mathbb E}[ \psi_i(b_\xi^\psi(x^{(1)}),\ldots, b_\xi^\psi(x^{(n)})) \mid x^{(i)}=(b_\xi^\psi)^{-1}(b) ]. 
\]
The affiliation condition on $G^\psi$ guarantees that bids under $\psi$ satisfy MLRP. Thus, the same argument as in the proof of Theorem~\ref{thm:main} shows that $b^{{\rm 1st}}_\xi (({b}^\psi_\xi)^{-1}(b))$ is increasing in $\xi$. \qed

\subsubsection{Proof of Lemma~\ref{lem:gencompetition}}\label{sec:gencompproof}

We start with two preliminary lemmas. 

\begin{lem}\label{lem:sc1}
Let $X$ be a measurable subset of $\mathbb R^d$ for some $d$, and assume that $(\theta, x) \in \Theta \times X$ are affiliated, with joint density $g(\theta, x)$ and conditional measures $\mu_\theta \in \Delta(X)$.  Suppose $H, L \subseteq X$ satisfy $H\geq L$, i.e., for all $h \in H$ and $\ell \in L$, we have $h  \geq \ell$.  Then
\[
 \mu_{\theta}(H) - \mu_{\theta} (L)>  0    \implies  \mu_{\theta'}(H) - \mu_{\theta'}(L) \geq 0 \ \forall \theta' > \theta.
\]
\end{lem}

\begin{proof}
Consider $\theta' > \theta$ and suppose $\mu_\theta(H) - \mu_\theta(L) > 0$. By affiliation, for all $h \in H, \ell \in L$, 
$
g(\theta', h) g(\theta, \ell) \geq  g(\theta, h)g(\theta', \ell) 
$.  Dividing both sides by the product of the marginal densities $g(\theta') g(\theta)$, we obtain 
$
g_{\theta'}(h) g_\theta(\ell) \geq g_\theta(h) g_{\theta'} (\ell).
$
This implies
\begin{align}
\mu_{\theta'}(H) \mu_\theta(L) \geq \mu_\theta(H)\mu_{\theta'} (L) .\label{eqn:1}
\end{align}

If $\mu_{\theta'}(L) = 0$, the desired conclusion is immediate.  Suppose $\mu_{\theta'}(L) > 0$.  Since $\mu_\theta(H) > \mu_\theta(L) \geq 0$, (\ref{eqn:1}) implies $\mu_\theta(L) > 0$.  Hence, by (\ref{eqn:1}),
$
\mu_{\theta'}(H) \geq \frac{\mu_{\theta}(H)}{\mu_{\theta}(L)} \mu_{\theta'}(L) \geq \mu_{\theta'}(L)
$. \end{proof}

\begin{lem}\label{lem:sc2}
Consider any two auctions with equilibrium bidding functions $(b_\xi)$ and $(\hat b_\xi)$ and bid distributions $G= (G_\theta)$ and $\hat{G}= (\hat G_\theta)$.  Suppose that 
\begin{enumerate}
\item $b_\xi (\hat{b}_\xi^{-1}(b))$ is increasing in $\xi$ for all $b$; and
\item $(\theta, \xi, b_\xi(v))$ are affiliated.
\end{enumerate}
Then whenever $\hat{G}_\theta(b) - G_\theta(h) > 0$, we have $\hat{G}_{\theta'}(b) - G_{\theta'}(h) \geq 0$ for all $\theta' > \theta$.
\end{lem}

\begin{proof} Write $h_\xi(b) = b_\xi (\hat{b}_\xi^{-1}(b))$ and consider any $\theta' > \theta$.  Letting $\mu_\theta$ denote the conditional measure over $(\xi, b_\xi (v))$, we have
\begin{align*}
\hat{G}_{\theta'}(b) - G_{\theta'}(h) = \mu_{\theta'} \left( b_\xi(v) \leq h_\xi(b)  \right) - \mu_{\theta'} \left( b_\xi(v) \leq h \right).
\end{align*}

If $h_\xi(b) \geq h$ for all $\xi$, then trivially $\hat{G}_{\theta'}(b) - G_{\theta'}(h) \geq 0$.  So suppose there exists some $\xi$ for which $h_\xi(b) \leq h$ so that 
$
\xi^* = \sup \left\{ \xi : h_{\xi'}(b) \leq h \ \forall \xi' \leq \xi \right\}
$ is well-defined. Letting\[
H = \left\{ (\xi, b') :  \xi \geq \xi^*, b' \in [h, h_\xi(b)) \right\}, \quad  L= \left\{ (\xi, b') : \xi < \xi^* , b' \in (h_\xi(b), h] \right\},
\]
note that 
\begin{equation*}
\hat{G}_{\theta'}(b)  - G_{\theta'}(h)  = \mu_{\theta'}\left( b_\xi(v) \leq h_\xi(b) \right) - \mu_{\theta'} \left(b_\xi(v) \leq h \right) = \mu_{\theta'} \left( (\xi, b_\xi(v)) \in H \right) - \mu_{\theta'} \left( (\xi, b_\xi(v)) \in L  \right).
\end{equation*}
Similarly, $0 < \hat{G}_\theta(b) - G_\theta (h) = \mu_{\theta} \left( (\xi, b_\xi(v)) \in H \right) - \mu_{\theta} \left( (\xi, b_\xi(v)) \in L  \right)$.
Since $H \geq L$,  Lemma~\ref{lem:sc1}  implies that $\hat{G}_{\theta'}(b) - G_{\theta'}(h) \geq 0$.
\end{proof}

To complete the proof of Lemma~\ref{lem:gencompetition}, we fix $\theta'>\theta$ and a bid $b$, and show that $h:=G_\theta^{-1}(\hat G_\theta(b)) \leq  h':=G_{\theta'}^{-1}(\hat G_{\theta'}(b))$. If $\hat G_\theta(b)=0$, then $h$ is the lower endpoint of the support of $G_\theta$. Affiliation of $(\theta,\xi,b_\xi(v))$ implies that the marginal family $(G_\theta)$ satisfies the MLRP, so its lower support endpoint is nondecreasing and $h\leq h'$.

Suppose next that $\hat G_\theta(b)>0$. If $h=+\infty$, then $\hat G_\theta(b)=1$. Applying Lemma~\ref{lem:sc2} along finite bids tending to infinity gives $\hat G_{\theta'}(b)=1$, while the FOSD ordering of the upper support endpoints implies that the support of $G_{\theta'}$ is also unbounded. Thus, $h'=+\infty$. Next, suppose that $h$ is finite. At every $x<h$ such that $G_\theta(x)<\hat G_\theta(b)$, Lemma~\ref{lem:sc2} gives $G_{\theta'}(x)\leq \hat G_{\theta'}(b)$. Hence, by continuity,
\[
G_{\theta'}\left(\sup\{x<h:G_\theta(x)<\hat G_\theta(b)\}\right)\leq \hat G_{\theta'}(b).
\]
If the supremum in this display equals $h$, it follows immediately that $G_{\theta'}(h)\leq \hat G_{\theta'}(b)$. Otherwise, $G_\theta$ is constant between that supremum and $h$, and the support of $G_\theta$ resumes at $h$. In this case, $G_{\theta'}$ is constant on the same interval: positive density under $\theta'$ inside this gap, compared with positive density under $\theta$ at a higher bid in its support, would violate the MLRP. Thus $G_{\theta'}(h)\leq \hat G_{\theta'}(b)$ in this case as well.

If $h$ belongs to the support of $G_{\theta'}$, the definition of the right inverse now gives $h'\geq h$. If it does not, since the upper support endpoints increase in the state,  the support of $G_{\theta'}$ resumes above $h$; continuity implies that its cdf at the first such support point is still at most $\hat G_{\theta'}(b)$, and again $h'\geq h$. \qed

\subsection{Proof of Proposition~\ref{prop:correlated}}\label{app:correlated}

Denote by ${G}^{\rm 2nd}_\theta (b' \mid b) = F_\theta \left( \left(b_\theta^{\rm 2nd} \right)^{-1}(b') \mid \left(b^{\rm 2nd}_\theta \right)^{-1}(b) \right) $
the cdf of $b'=\max_{j\not=i}b^{\rm 2nd}_\theta(v_j)$ conditional on $b^{\rm 2nd}_\theta(v_i) = b$, and by $g^{\rm 2nd}_\theta( \cdot \mid b)$ its density. For each $b'\leq b$, define
\[
\hat{L}_\theta (b' \mid b) = L_\theta ((b^{\rm 2nd}_\theta)^{-1}(b') \mid (b^{\rm 2nd}_\theta)^{-1}(b)) = e^{- \intop_{b'}^{b} \frac{g^{\rm 2nd}_\theta(x \mid x)}{{G}^{\rm 2nd}_\theta(x \mid x)} dx}.
\]
  Given this, we can write $
b_\theta^{\rm 1st} \left( \left(b_\theta^{\rm 2nd} \right)^{-1}(b) \right) = \intop_{0}^{b} b' \, d\hat{L}_\theta (b' \mid  b) \label{eqn:2}$.

By assumption, $(\theta,b_1^{\rm 2nd},\ldots,b_n^{\rm 2nd})$ is
affiliated, and its bid coordinates are symmetric. Hence,
\citet[Theorem 2]{milgrom1982} implies that, for every $i$,
$
(\theta,b_i^{\rm 2nd},\max_{j\neq i}b_j^{\rm 2nd})
$
is affiliated. It follows that, for any $b$, $G^{\rm 2nd}_\theta(\cdot\mid b)$ is increasing in $\theta$ in the reverse-hazard-rate order. Therefore $\hat L_\theta(\cdot\mid b)$ is FOSD-increasing in $\theta$, and $b_\theta^{\rm 1st}((b_\theta^{\rm 2nd})^{-1}(b))$ is increasing in $\theta$  when $(b_\theta^{\rm 2nd})^{-1}(b)\in (\underline v_\theta,\overline v_\theta)$. The boundary cases follow by observing that affiliation orders the support endpoints,  and $b_\theta^{\rm 1st}((b_\theta^{\rm 2nd})^{-1}(b))\leq b$, with equality at the lower endpoint.
 Lemma~\ref{lem:competition} then implies that $G^{\rm 1st}$ is more accurate than $G^{\rm 2nd}$. \qed

\subsection{Proof of Proposition~\ref{prop:multiple}}
Since $(F_\theta)$ satisfies the MLRP, the distribution of $v_{-i}^{(m)}$ also satisfies the MLRP \citep[Theorem 1.C.33]{shaked2007}. For every fixed $v$, upper truncation to $(-\infty,v]$ preserves the MLRP, and therefore
$
b^{\rm disc}_\theta(v)={\mathbb E}_\theta[v_{-i}^{(m)}\mid v_{-i}^{(m)}\leq v]
$
is increasing in $\theta$. Since $b^{\rm vcg}_\theta(v)=v$ is state-independent, $b_\theta^{\rm disc}((b_\theta^{\rm vcg})^{-1}(b))$ is increasing in $\theta$ when $(b_\theta^{\rm vcg})^{-1}(b)\in (\underline v_\theta, \overline v_\theta)$. The boundary cases follow by observing that the MLRP orders the support endpoints,  and $b_\theta^{\rm disc}(v)\leq v$, with equality at the lower endpoint.
By Lemma~\ref{lem:competition}, $G^{\rm disc}$ is more accurate than $G^{\rm vcg}$. \qed

\subsection{Details for Remark~\ref{rem:MLRP} on the MLRP of $(G^{\psi}_\theta)$}\label{app:MLRP} 

We present two natural settings in which $(G^{\psi}_\theta)$ satisfies the MLRP: $k$th-price and all-pay auctions with power-distribution values; and $k$th-price auctions with a sufficiently large number of bidders under all value distributions $(F_\theta)$ satisfying some technical conditions.

\subsubsection{Parametric Environment}\label{app:MLRP-parametric}

Suppose each value distribution $F_\theta$ is a power distribution, i.e., $F_\theta(v)=(v/\alpha_\theta)^{\beta_\theta}$, for some $\alpha_\theta, \beta_\theta>0$. Then the support of $F_\theta$ is $[0, \alpha_\theta]$.  We assume that both $\alpha_\theta$ and $\beta_\theta$ are increasing in $\theta$, which is equivalent to the MLRP of $(F_\theta)$. 

As the calculations below verify, $(G^{k{\rm th}}_\theta)$ with $k\leq 2$ and $(G^{\rm all}_\theta)$ satisfy the MLRP.  
Moreover, for any $k \geq 3$, $(G^{k{\rm th}}_\theta)$ satisfies the MLRP if either (i) $\beta_\theta$ is independent of $\theta$, or (ii) for each pair of states $\theta'>\theta$, $\frac{\alpha_{\theta'}}{\alpha_{\theta}}$ is sufficiently large given $\beta_\theta,\beta_{\theta'}$. 

To illustrate the role of increasing $\alpha_\theta$, recall that 3rd-price bids $b^{\rm 3rd}_\theta (v)$ are decreasing in $\theta$. If $\alpha_\theta$ is constant while $\beta_\theta$ is strictly increasing, then the highest possible bid $b^{\rm 3rd}_\theta(\alpha_\theta)$ is strictly decreasing in $\theta$, which leads to an MLRP violation (however, even in this case, $G^{\rm 1st}$ dominates $G^{\rm 3rd}$ by Proposition~\ref{prop:example}).  Note that it is natural in the current setting to assume that the highest possible value $\alpha_\theta$ is strictly increasing in $\theta$. Indeed, under the interpretation of our model where values are functions $v_i=u(\theta, \varepsilon_i)$ of a common fundamental $\theta$ and state-independent idiosyncratic taste shock $\varepsilon_i$, this is the case as long as $u$ is strictly increasing in $\theta$.

\medskip

\noindent{\bf Calculations.} As shown by \cite{mihelich2020}, the $k$th-price equilibrium bidding function is linear in types, i.e., $b^{k{\rm th}}_\theta(v)=\delta_\theta v$, where 
\[
\delta_\theta=\frac{\Gamma(n-k+1)\Gamma(n-1+1/\beta_\theta)}{\Gamma(n-1)\Gamma(n-k+1+1/\beta_\theta)}
\]
and $\Gamma$ is the Gamma function. 
The corresponding bid distribution is given by $G^{k{\rm th}}_\theta(b)=\left(\frac{b}{\delta_\theta\alpha_\theta}\right)^{\beta_\theta}$, which is supported on $[0, \delta_\theta\alpha_\theta]$ and admits the density $g_\theta(b)=(\delta_\theta\alpha_\theta)^{-\beta_\theta}\beta_\theta b^{\beta_\theta-1}$. 
Thus, for all states $\theta'>\theta$ and bids $b'>b>0$ that belong to $[0, \delta_\theta\alpha_\theta]\cap [0, \delta_{\theta'}\alpha_{\theta'}]$,
\[
\frac{g^{k{\rm th}}_\theta(b')}{g^{k{\rm th}}_\theta(b)}=(b'/b)^{\beta_\theta-1} \leq (b'/b)^{\beta_{\theta'}-1}=\frac{g^{k{\rm th}}_{\theta'}(b')}{g^{k{\rm th}}_{\theta'}(b)},
\]
consistent with the MLRP. Thus, the MLRP holds if and only if the highest bid in the support $\delta_\theta\alpha_\theta$ is increasing in $\theta$. 
This is the case for $k=1$ (where $\delta_\theta=\frac{(n-1)\beta_\theta}{(n-1)\beta_\theta+1}$) and $k = 2$ (where $\delta_\theta=1$). 
For $k\geq 3$ and for all $\theta'>\theta$, we have $\delta_{\theta'}\alpha_{\theta'}\geq\delta_\theta\alpha_\theta$ whenever (i) $\beta_\theta$ is independent of $\theta$ or (ii) $\frac{\alpha_{\theta'}}{\alpha_{\theta}}$ is sufficiently large.  

Finally, under the all-pay auction, the equilibrium bidding function is
\[
b^{\rm all}_\theta(v)=(v/\alpha_\theta)^{\beta_\theta(n-1)} b_\theta^{\rm 1st}(v)=\gamma_\theta v^{\beta_\theta(n-1)+1},
\quad \text{ where } \quad
\gamma_\theta=\frac{(n-1)\beta_\theta}{(n-1)\beta_\theta+1}\alpha_\theta^{-\beta_\theta(n-1)}.
\]
The corresponding bid cdf and density are
\begin{equation*}
G^{\rm all}_\theta(b)
=\alpha_\theta^{-\beta_\theta}
\left(\frac{b}{\gamma_\theta}\right)^{\frac{\beta_\theta}{\beta_\theta(n-1)+1}}, \quad \quad
g^{\rm all}_\theta(b)
=\frac{\alpha_\theta^{-\beta_\theta}\beta_\theta}{\beta_\theta(n-1)+1}
\left(\frac{b}{\gamma_\theta}\right)^{\frac{\beta_\theta}{\beta_\theta(n-1)+1}-1}
\gamma_\theta^{-1},
\end{equation*}
with support $[0,\gamma_\theta\alpha_\theta^{\beta_\theta(n-1)+1}]$. Thus, for $\theta'>\theta$ and $b'>b>0$ in the common support,
\begin{equation*}
\frac{g^{\rm all}_\theta(b')}{g^{\rm all}_\theta(b)}
=(b'/b)^{\frac{\beta_\theta}{\beta_\theta(n-1)+1}-1}
\leq (b'/b)^{\frac{\beta_{\theta'}}{\beta_{\theta'}(n-1)+1}-1}
=\frac{g^{\rm all}_{\theta'}(b')}{g^{\rm all}_{\theta'}(b)}.
\end{equation*}
The highest bid in the support of $g_\theta$ is  $\gamma_\theta\alpha_\theta^{\beta_\theta(n-1)+1}=\frac{(n-1)\beta_\theta}{(n-1)\beta_\theta+1}\alpha_\theta$, which is increasing in $\theta$. Thus, the MLRP holds.

\subsubsection{Large Number of Bidders}\label{app:MLRP-large-n} 
We show that the bid distributions in the $k$th-price auction satisfy the MLRP when $n$ is sufficiently large, under additional technical conditions on $(F_\theta)$. This is immediate for $k=2$, so we focus on $k\geq 3$. 

\begin{lem}\label{lem:MLRP-large-n}
Fix $k\geq 3$. Suppose that:
\begin{enumerate}
\item $\Theta$ is finite;
\item the support of $F_\theta$ is $I_\theta=[\underline v,\overline v_\theta]$, with a common lower endpoint $\underline v$ and $\overline v_{\theta'}>\overline v_\theta$ whenever $\theta'>\theta$;
\item each $F^{-1}_\theta$ is $k$-times continuously differentiable, with $\min_{x\in [0,1]}(F^{-1}_\theta)'(x)>0$;
\item for any pair $\theta'>\theta$, 
$
\inf_{v\in I_\theta}
\frac{d}{dv}\log\left(\frac{f_{\theta'}(v)}{f_\theta(v)}\right)>0.
$
\end{enumerate}
Then there exists $N>k$ such that, for every $n\geq N$, the $k$th-price auction admits a unique equilibrium and its equilibrium bid distributions $(G^{k{\rm th}}_{\theta})$ satisfy the MLRP.
\end{lem}

\begin{proof} Write bidding functions as functions of value quantiles $x \in [0, 1]$. By \cite{mihelich2020},  a symmetric monotone $k$th-price equilibrium exists and is uniquely given by
\begin{equation}\label{eq:k-large-n-quantile}
b_{\theta}^{k{\rm th}}(x)
=F^{-1}_\theta(x)+\sum_{j=1}^{k-2}
\binom{k-2}{j}\frac{(n-k)!}{(n-k+j)!}
 x^j (F^{-1}_\theta)^{(j)}(x)
\end{equation}
if and only if the expression in (\ref{eq:k-large-n-quantile}) is strictly increasing in $x$.
The coefficient multiplying the $j$th term is $O(n^{-j})$.  Thus, by Assumptions 1 and 3,
\begin{equation}\label{eq:k-large-n-c2}
\max_{\theta\in\Theta}
\left\|b_{\theta}^{k{\rm th}}-F^{-1}_\theta\right\|_{C^2}
=O(n^{-1}),
\end{equation}
where $\| \cdot \|_{C^2}$ denotes the $C^2$-norm.
Since $\min_{\theta,x}(F^{-1}_\theta)'(x)>0$, \eqref{eq:k-large-n-c2} implies that $b_{\theta}^{k{\rm th}}$ is indeed strictly increasing for every $\theta$ once $n$ is sufficiently large. Hence, \eqref{eq:k-large-n-quantile} defines the unique symmetric monotone equilibrium.

The lower endpoint of the bid support is $b_{\theta}^{k{\rm th}}(0)=\underline v$, while
$
b_{\theta}^{k{\rm th}}(1)
=\overline v_\theta+O(n^{-1}).
$
Because $\Theta$ is finite and $\overline v_\theta$ is strictly increasing in $\theta$, $b_{\theta}^{k{\rm th}}(1)$ is also strictly increasing in $\theta$ for all sufficiently large $n$. Fix $\theta'>\theta$.  For large enough $n$, the bid support in state $\theta$ is contained in that in state $\theta'$, so we can define
\[
R_{\theta',\theta}(x)
=G^{k{\rm th}}_{\theta'}(b^{k{\rm th}}_{\theta}(x))
=(b^{k{\rm th}}_{\theta'})^{-1}(b^{k{\rm th}}_{\theta}(x)),
\qquad \forall \, x\in[0,1].
\]
Equation~\eqref{eq:k-large-n-c2}, together with $\min_{\theta,x}(F^{-1}_\theta)'(x)>0$, implies $\lim_{n\to\infty}\|R_{\theta',\theta}-F_{\theta'}\circ F^{-1}_{\theta}\|_{C^2}=0$.

Observe that 
\[
\frac{d}{dx} (F_{\theta'}\circ F^{-1}_{\theta})(x)
=\frac{f_{\theta'}(F^{-1}_\theta(x))}{f_\theta(F^{-1}_\theta(x))},
\]
\[
\frac{d^2}{dx^2} (F_{\theta'}\circ F^{-1}_{\theta})(x)
=(F^{-1}_\theta)'(x)
\frac{f_{\theta'}(F^{-1}_\theta(x))}{f_\theta(F^{-1}_\theta(x))}
\left.
\frac{d}{dv}\log\left(\frac{f_{\theta'}(v)}{f_\theta(v)}\right)
\right|_{v=F^{-1}_\theta(x)}.
\]
By Assumptions~1--4, there is $\eta>0$ with
$\frac{d^2}{dx^2} (F_{\theta'}\circ F^{-1}_{\theta})(x)\geq\eta$ for every $x$ and all pairs $\theta'>\theta$.  Thus, the $C^2$-convergence of $R_{\theta',\theta}$ to $F_{\theta'}\circ F^{-1}_{\theta}$ implies that
$R_{\theta',\theta}''(x)\geq\eta/2$ for all $\theta' > \theta$ and all $x$ once $n$ is sufficiently large. 

Finally, if $g^{k{\rm th}}_{\theta}$ denotes the equilibrium bid density, then $R_{\theta',\theta}'(x)
=g^{k{\rm th}}_{\theta'}(b^{k{\rm th}}_{\theta}(x))(b^{k{\rm th}}_{\theta})'(x)
=\frac{g^{k{\rm th}}_{\theta'}(b^{k{\rm th}}_{\theta}(x))}
{g^{k{\rm th}}_{\theta}(b^{k{\rm th}}_{\theta}(x))}$. Thus, convexity of $R_{\theta',\theta}$ yields monotonicity of $\frac{g^{k{\rm th}}_{\theta'}(b^{k{\rm th}}_{\theta}(x))}
{g^{k{\rm th}}_{\theta}(b^{k{\rm th}}_{\theta}(x))}$ on the common support. By the monotonicity of the bid-support endpoints $b_{\theta}^{k{\rm th}}(1)$ established above, this proves the MLRP. 
\end{proof}

\footnotesize
\begin{spacing}{0.01}
\bibliographystyle{econometrica}
\bibliography{auction}
\end{spacing}

\newpage
\setcounter{page}{1}

\singlespacing

\begin{center}
 {\Large\textbf{Online Appendix to ``Auctions as Experiments''\\[0.7cm]}}
 {\large Mira Frick, Ryota Iijima, Yuhta Ishii, and Nicholas Wu \\[1cm]}
\end{center}

\normalsize

\section{Payoff Comparison under Non-i.i.d.\ Observations}

\subsection{Extension of Lehmann to Multi-dimensional Signals}\label{app:general lehmann}

This section provides an extension of \cites{lehmann1988} payoff comparison result to multi-dimensional signals.   This allows us to extend the payoff comparison in Corollary 1 to certain settings where observed bids are not i.i.d.\ conditional on $\theta$; for example, if only the $k$ highest bids are observed, or buyers' types are correlated conditional on $\theta$.

We take a compact state space $\Theta\subseteq \mathbb{R}$, and compare a pair of experiments $H^i=(H^i_\theta)$,   $i=1,2$, where $H^i_\theta\in\Delta(X^i)$ is the signal distribution conditional on each state $\theta$, which admits a corresponding density $h^i_\theta$ with respect to some underlying  measure on $X^i$. 
Each signal space $X^i$ is a measurable lattice endowed with partial order $\succsim^i$. 

As in Section~\ref{sec:payoff}, the DM faces a monotone decision problem ${\cal D}=(A, u)$, where  $A\subseteq\mathbb R$ is an action set that is either finite or a compact interval, and $u: A\times\Theta\to\mathbb R$ is a continuous utility function that is single-peaked in $a$ at all $\theta$, with a maximizer $a^* (\theta) = \argmax_{a \in A} u(a, \theta)$ that is increasing in $\theta$.
Let $U({\cal D}, H^i)$ denote the DM's expected payoff at ${\cal D}$ when he observes a signal from $H^i$:  
\[
U({\cal D}, H^i)=\max_{\sigma^i: X^i \to A}{\mathbb E}[u(\sigma^i(x),\theta)]
\] 
where the expectation is with respect to the state (which is drawn from the prior $q_0 \in \Delta(\Theta)$) and the signal (which is drawn from $H^i_\theta$ at each $\theta$).

We say that a collection of measurable functions $(\phi_\theta)$, where $ \phi_\theta : X^2 \rightarrow  X^1$ for each $\theta$, is a \textbf{\textit{monotone family}} if the following conditions hold:

\begin{enumerate}
\item  Each $\phi_\theta$ is an order embedding, i.e., there exists a measurable set $S_\theta\subseteq X^2$ with $H_\theta^2(S_\theta)=1$ such that 
$
\phi_\theta(x) \succsim^1 \phi_\theta(z) \Rightarrow x\succsim^2 z
$ for every $z\in X^2$ and $x\in S_\theta$.  

\item $\phi_\theta$ is monotone in $\theta$, i.e., 
$
\theta \geq \theta' \Rightarrow \phi_{\theta}(x) \succsim^1 \phi_{\theta'}(x)
$ for every $x\in X^2$.

\item For each measurable $E\subseteq X^2$ and $\theta$, the upper closure $\mathord{\uparrow}\phi_\theta(E)$ is measurable.\footnote{Recall that $E \subseteq X^i$ is an \textit{\textbf{upper set}} if $x \in E$ and $x' \succsim^i  x$ imply $x' \in E$. For $E \subseteq X^i$, let
$
\mathord{\uparrow}E=\{z\in X^i:z\succsim^i d\text{ for some }d\in E\}
$ denote its \textit{\textbf{upper closure}}, which is an upper set by definition. } 
\end{enumerate}

The following result provides conditions under which  $H^1$ yields a higher value of information than $H^2$ under all monotone decision problems. 

\begin{thm}\label{thm:multidimension}
Suppose that (i) there is a monotone family $(\phi_\theta)$ such that $H^1_\theta(E)=H^2_\theta(\phi_\theta^{-1}(E))$ for each $\theta$ and measurable $E\subseteq X^1$, and 
(ii) $H^2$ satisfies affiliation.\footnote{That is, $
 h^2( x \wedge x' \mid  \theta') h^2( x \vee x' \mid  \theta) \geq h^2(x  \mid  \theta)h^2(x' \mid  \theta')$ for every $x,x'\in X^2$ and $\theta>\theta'$.} 
 Then $U({\cal D}, H^1)\geq U({\cal D}, H^2)$
for any monotone decision problem $\cal D$.
\end{thm}

Condition (i) requires that a signal observation from $H^1_\theta$ is equivalent to observing $\phi_\theta(x)$ where $x$ is drawn from $H^2_\theta$ and $(\phi_\theta)$ is a monotone family. This requirement can be seen as a generalization of the condition used in Lemma~\ref{lem:competition}. 
Under the baseline model with i.i.d.\ values (Section~\ref{sec:model}), each experiment corresponds to the bid distributions in each state $\theta$.  To see the connection, take two auctions and their corresponding equilibrium bidding functions $b$ and $\hat b$ such that $b_\theta ((\hat b_\theta^{-1}) (\cdot))$ is increasing in $\theta$ as required in Lemma~\ref{lem:competition}. Then, assuming that all $n$ bids are observed by the DM, the corresponding monotone family takes the form $\phi_\theta(b_1,\ldots, b_n)=(b_\theta (\hat b_\theta^{-1} (b_1)),\ldots, b_\theta(\hat b_\theta^{-1} (b_n)))$ for possible bids $(b_1,\ldots, b_n)$ under $\hat b_\theta$ .

\begin{rem}[Comparison with \cite{di2021}]
The result is complementary to \cite{di2021}, who provide an alternative  multi-dimensional signal extension. They consider a pair of experiments $H^i$, $i=1,2$, defined on signal space $X^1=X^2=\mathbb R^n$, both of which satisfy affiliation and admit continuous densities with respect to the Lebesgue measure.  
Their Theorem 0 shows that  $U({\cal D}, H^1)\geq U({\cal D}, H^2)$ holds for any monotone decision problem $\cal D$ if each $H^2_\theta$ can be obtained from $H^1_\theta$ by a so-called Knothe-Rosenblatt rearrangement $\zeta_\theta: X^1\to X^2$ that satisfies $\zeta_\theta(x)\leq\zeta_{\theta'}(x)$ for each $x\in X^1$ and $\theta>\theta'$.\footnote{Formally, $\zeta_\theta$ is defined by $\zeta_\theta(x_1,\ldots, x_n)=(\hat x_1,\ldots, \hat x_n)$,
where $ \hat x_1= (H^2_{\theta,1})^{-1}(H^1_{\theta, 1}(x_1))$ and $\hat x_k= (H^2_{\theta,k, \hat x_{k-1},\ldots, \hat x_1})^{-1}(H^1_{\theta, k, x_{k-1},\ldots, x_1}(x_k))
$ for each $k=2,\ldots, n$. Here, $H^1_{\theta, 1}$ (resp. $H^2_{\theta, 1}$) denotes $H^1_\theta$'s (resp. $H^2_\theta$'s) marginal cdf on the first coordinate, and $H^1_{\theta, k, x_{k-1},\ldots, x_1}$ (resp. $H^2_{\theta,k, \hat x_{k-1},\ldots, \hat x_1}$) denotes $H^1_\theta$'s (resp. $H^2_\theta$'s) marginal cdf on the $k$th coordinate conditional on $(x_{k-1},\ldots, x_1)$ (resp. conditional on $(\hat x_{k-1},\ldots, \hat x_1)$). }
Thus, the result in \cite{di2021} uses a specific way $(\zeta_\theta)$ of relating $H^1_\theta$ and $H^2_\theta$ on a common finite-dimensional signal space, while Theorem~\ref{thm:multidimension} allows for abstract mappings $(\phi_\theta)$ defined on general signal spaces that are required to be order embeddings.  Neither condition nests the other in general (even if we focus on the case $X^1=X^2=\mathbb R^n$). \finex
\end{rem}

\subsection{Applications to Auction Bid Observations}

Theorem~\ref{thm:multidimension} can be applied to settings where Corollary~\ref{cor:lehmann} does not apply but it is straightforward to construct corresponding monotone families.

\

\noindent{\bf Partial bid observations:} Consider an i.i.d.-type setting as in Section~\ref{sec:model}, Section~\ref{sec:interdependent}, or the multi-unit setting in Section~\ref{sec:other-extensions}. 
Suppose that the DM observes a subset of order statistics of bids, where the selection does not depend on the auction format $\psi$. That is, we fix $K\subseteq\{1,\ldots,n\}$ and $H^\psi_\theta$ is the distribution of $(b^{(k)})_{k\in K}$ at $\theta$, where $(b^{(1)},...,b^{(n)})$ are the order statistics of the bids in auction $\psi$.  For example,  the DM observes the $k$ highest bids if $K=\{1,\ldots, k\}$.

\begin{cor}\label{cor:partial} Consider an i.i.d. type setting. Take two auction formats $\psi,\hat\psi$ such that $G^{\psi}$ is more accurate than $G^{\hat\psi}$ and the latter satisfies the MLRP. Then for any $K\subseteq \{1,\ldots, n\}$,   $U({\cal D}, H^\psi)\geq U({\cal D}, H^{\hat\psi})$ holds for any monotone decision problem $\cal D$.
\end{cor}

\begin{proof}
Since $G^{\hat\psi}$ satisfies the MLRP,   $H^{\hat\psi}$ satisfies affiliation.  
Let $b\in\mathbb R_+^{|K|}$ and $\hat b\in\mathbb R_+^{|K|}$ denote signal realizations under $H^{\psi}$ and $H^{\hat\psi}$, respectively.  
Take $\phi_\theta$ so that for each $k\in K$, the corresponding coordinate is $\phi_{\theta,k}(b)=b^\psi_\theta((b_\theta^{\hat\psi})^{-1}(b_k))$  for possible bids $(b_1,\ldots, b_n)$ under $b^{\hat\psi}_\theta$. 
Then each $\phi_\theta$ is an order embedding, since the bidding function is increasing. It also satisfies monotonicity by Lemma~\ref{lem:competition}. 
By construction $H^\psi_\theta(E)=H^{\hat\psi}_\theta(\phi^{-1}_\theta(E))$ for each measurable $E$, so that the conclusion follows from Theorem~\ref{thm:multidimension}. 
\end{proof}

\noindent{\bf Imperfect state observation:}
Consider the setting in Section~\ref{sec:imperfect}, where buyers observe a common imperfect signal $\xi$ about $\theta$. 
We compare the first-price auction against a general auction $\psi$. Allowing  for partial bid observations as above, fix $K\subseteq \{1,\ldots, n\}$, and let $H^{\rm 1st}_\theta$ (resp. $H^{\psi}_\theta$) denote the distribution of $(b^{(k)})_{k\in K}$ at $\theta$, where $(b^{(1)},...,b^{(n)})$ are the order statistics of the bids in the first-price auction (resp. auction $\psi$). 
Theorem~\ref{thm:multidimension} does not directly apply to this setting; this is because we may not be able to construct a monotone family, since bids are functions of $\xi$ not $\theta$. Therefore, in addition to the affiliation condition in Proposition~\ref{prop:imperfect}, we require that $(v_i)$ and $\theta$ are independent conditional on $\xi$. We also focus on \textit{\textbf{concave}} monotone decision problems, where $A$ is an interval and each $u(\cdot, \theta)$ is strictly concave.  

\begin{cor}\label{cor:xi comparison} Consider the setting with imperfect state observation in Section~\ref{sec:imperfect}. Suppose that (i) $(\theta, \xi, (b_i))$ is affiliated under both the first-price auction and some auction $\psi$, and (ii) $(v_i)$ and $\theta$ are independent conditional on $\xi$.  Then for any $K\subseteq \{1,\ldots, n\}$,   $U({\cal D}, H^{\rm 1st})\geq U({\cal D}, H^{\psi})$ holds for any concave monotone decision problem $\cal D$.
\end{cor}

\begin{proof}
The DM's problem can be recast as a modified decision problem in which $\xi$ is the state and the utility is given by $\hat u(a, \xi)={\mathbb E}[u(a,\theta)|\xi]$. By strict concavity of $u(\cdot, \theta)$,  $\hat u(\cdot, \xi)$ is also strictly concave, and thus  single-peaked.   
Since $u(\cdot, \theta)$ is single peaked with its peak increasing in $\theta$, and $(\theta, \xi)$ is affiliated, Theorem 2 of \cite{quah2009} ensures that the maximizer of $\hat u(\cdot, \xi)$ is increasing. Therefore, the modified decision problem is monotone, and the conclusion follows from Proposition~\ref{prop:imperfect} and Corollary~\ref{cor:partial}.
\end{proof}

\noindent{\bf Correlated values:} Consider the correlated-value setting in Section~\ref{sec:other-extensions}. As above, we allow for partial bid observations. 
Fix $K\subseteq \{1,\ldots, n\}$, and let $H^{\rm 1st}_\theta$ (resp. $H^{\rm 2nd}_\theta$) denote the distribution of $(b^{(k)})_{k\in K}$ at $\theta$, where $(b^{(1)},...,b^{(n)})$ are the order statistics of the bids in the first-price auction (resp. the second-price auction). 

\begin{cor}\label{cor:correlated} Consider the correlated-type setting in Section~\ref{sec:other-extensions}. Assume $(\theta, (b^{\rm 2nd}_\theta(v_i))_{i=1}^n)$ is affiliated. 
Then for any $K\subseteq \{1,\ldots, n\}$,   $U({\cal D}, H^{\rm 1st})\geq U({\cal D}, H^{\rm 2nd})$ holds for any monotone decision problem $\cal D$.
\end{cor}
\begin{proof}
The affiliation and exchangeability  of the second-price auction bids ensure that  $H^{\rm 2nd}$ satisfies affiliation.  
As in the proof of Corollary~\ref{cor:partial}, we consider the mapping $\phi_\theta$ such that $\phi_{\theta,k}(b)=b^{\rm 1st}_\theta((b_\theta^{\rm 2nd})^{-1}(b_k))$ for each $k\in K$  for possible bids $(b_1,\ldots, b_n)$ under $b^{\rm 2nd}_\theta$. Then each $\phi_\theta$ is an order embedding, and it also satisfies monotonicity by the proof of  Proposition~\ref{prop:correlated}. 
By construction $H^{\rm 1st}_\theta(E)=H^{\rm 2nd}_\theta(\phi^{-1}_\theta(E))$ for each measurable $E$, so that the conclusion follows from  Proposition~\ref{prop:correlated} and Theorem~\ref{thm:multidimension}. 
\end{proof}

\subsubsection{Proof of Theorem~\ref{thm:multidimension}}

We start with two preliminary lemmas.

\begin{lem}\label{lem:monotone-phi}
For any upper set $E\subseteq X^2$ and any $\theta^* \in \Theta$, write $E^*=\mathord{\uparrow}\phi_{\theta^*}(E)$.
Then
$
\left( H^1_{\theta} (E^*) - H^2_{\theta} \left(E \right) \right)(\theta - \theta^*) \geq 0
$  for every $\theta$. Moreover, $H^1_{\theta^*}(E^*)=H^2_{\theta^*}(E)$.
\end{lem}

\begin{proof}
Note that $E^*$ is an upper set in $X^1$ by construction.  If $\theta\geq \theta^*$ and $x\in E$, then
$
\phi_\theta(x)\succsim^1 \phi_{\theta^*}(x)
$
by monotonicity, so $\phi_\theta(x)\in E^*$.  Hence $H^1_\theta(E^*)\geq H^2_\theta(E)$.  

Conversely, suppose that $\theta\leq \theta^*$,
$x\in S_\theta$, and $\phi_\theta(x)\in E^*$.
There is then some $z\in E$ such that
\[
\phi_\theta(x)\succsim^1\phi_{\theta^*}(z)
\succsim^1\phi_\theta(z),
\]
where the second relation follows by monotonicity.
Since $\phi_\theta$ is an order embedding, $x\succsim^2 z$.
Since $E$ is an upper set, $x\in E$.
As $H_\theta^2(S_\theta)=1$, it follows that
$H^1_\theta(E^*)\leq H^2_\theta(E)$.
\end{proof}

\begin{lem}\label{lem:nested-upsets}
For 
 any pair of upper sets $E\supseteq E'$ in $X^2$ and $\theta<\theta'$, we have $\mathord{\uparrow}\phi_{\theta}(E)\supseteq 
\mathord{\uparrow}\phi_{\theta'}(E')$.  
\end{lem}

\begin{proof}
If $z\in\mathord{\uparrow}\phi_{\theta'}(E')$, there is some $x\in E'$ such that
\[
z\succsim^1 \phi_{\theta'}(x)\succsim^1 \phi_{\theta}(x),
\]
where the second relation follows by monotonicity. Since $E'\subseteq E$, this implies $z\in\mathord{\uparrow}\phi_{\theta}(E)$.
\end{proof}

\begin{proof}[Proof of Theorem~\ref{thm:multidimension}]

We first consider any decision problem $\cal D$ with a finite action set, which we write as $A=\{a_1, \cdots, a_J\}$ with $a_1<\cdots<a_J$.  Since the utility function has a single peak that is increasing in the state, there exist $\theta_1\leq\cdots\leq\theta_{J-1}$
such that, writing
\[
d_j(\theta)=u(a_{j+1},\theta)-u(a_j,\theta)
\]
for each $\theta\in\Theta$,  we have $d_j(\theta)\leq0$ for $\theta<\theta_j$ and
$d_j(\theta)\geq0$ for $\theta>\theta_j$.  Affiliation implies that an optimal strategy $\hat\sigma$ under $H^2$ can be chosen to be monotone, i.e., $\hat\sigma(x)\geq\hat\sigma(x')$ for any $x\succsim^2 x'$.   This is because $x\succsim^2 x'$ ensures that the DM's posterior at $x$ dominates the posterior at $x'$ in the likelihood-ratio order. Let $\hat\sigma$ denote such an optimal rule and define
\[
E_{j+1}=\{x\in X^2:\hat\sigma(x)\geq a_{j+1}\},\qquad j=1,\ldots,J-1.
\]
The sets $E_2\supseteq\cdots\supseteq E_J$ are upper sets.  For each $j$, set
\[
E'_{j+1}=\mathord{\uparrow}\phi_{\theta_j}(E_{j+1}).
\]
Lemma~\ref{lem:nested-upsets} implies that $E'_2\supseteq\cdots\supseteq E'_J$, so these sets define a strategy $\sigma$ under $H^1$: choose $a_1$ outside $E'_2$, choose $a_j$ on $E'_j\setminus E'_{j+1}$ for $1<j<J$, and choose $a_J$ on $E'_J$.

The payoff from any action $a_\ell$ admits a binary decomposition
\[
u(a_\ell,\theta)=u(a_1,\theta)+\sum_{j=1}^{J-1}
\mathbf{1}\{a_\ell\geq a_{j+1}\}d_j(\theta).
\]
Consequently, in every state $\theta$,
\begin{align*}
&\int u(\sigma(x),\theta)dH^1_\theta(x)
-\int u(\hat\sigma(x),\theta)dH^2_\theta(x) \\
&\qquad=\sum_{j=1}^{J-1}d_j(\theta)
\left[H^1_\theta(E'_{j+1})-H^2_\theta(E_{j+1})\right]\geq0.
\end{align*}
Here, every summand is nonnegative by Lemma~\ref{lem:monotone-phi}, applied with
$\theta^*=\theta_j$.  Integrating over the prior shows that $U({\cal D}, H^1)\geq U({\cal D}, H^2)$.

For decision problems in which $A$ is a compact interval, the conclusion follows from an approximation argument as in  \cite{lehmann1988}. 
\end{proof}

\section{Entry Fees}\label{app:fee}

Consider a standard auction $\psi$ combined with an entry fee $p\geq 0$. Each buyer must pay $p$ to participate in the auction. The entry decision and the bid are chosen simultaneously.  
The item is allocated to the highest bidder among the participants. The participant with the $i$th-highest bid pays the auction payment $\psi_i(b^{(1)},\ldots,b^{(n)})$, in addition to the entry fee, where the bids of nonparticipants are recorded as zero.   For simplicity, we focus on the case in which $p< \overline v_\theta$ for each state $\theta$.  

We consider symmetric monotone cutoff equilibria: In each state $\theta$, all types $v\geq c_\theta$ participate and use a strictly increasing and continuous bidding strategy, whereas types $v<c_\theta$ do not participate. Here, the cutoff $c_\theta<\overline v_\theta$ is independent of the auction format and uniquely determined by
\[
 c_\theta F_\theta(c_\theta)^{n-1}=p, 
\]
since $ c_\theta F_\theta(c_\theta)^{n-1}$ is the equilibrium gross payoff of type $c_\theta$ from entering the auction.
For example, the first-price, second-price, and all-pay auctions admit unique equilibria given by  
\[
 b_\theta^{{\rm 1st},p}(v)=
 \mathbb E_\theta\!\left[v^{(2)}\mathbf 1_{\{v^{(2)} \geq c_\theta\}}\mid v=v^{(1)}\right], \; \quad
 b_\theta^{{\rm 2nd},p}(v)=v, \;  \quad b_\theta^{{\rm all},p}(v)=F_\theta(v)^{n-1} b_\theta^{{\rm 1st},p}(v)
\]
for $v\geq c_\theta$.  Let $G^{\psi,p}_\theta$ denote the equilibrium bid distribution under auction format $\psi$ and entry fee $p$. The following result extends Theorem~\ref{thm:main} and provides an analog of Proposition~\ref{prop:reserve}.

\begin{prop}\label{prop:entry}
For any $p$, $G^{{\rm 1st},p}$ is more accurate than $G^{\psi,p}$, provided that $(G^{\psi,p}_\theta)$ satisfies the MLRP. 
\end{prop}

\begin{proof} Note that, for every $b\in\mathbb R_+$, we obtain 
$
(G^{{\rm 1st}, p}_\theta)^{-1}\!\left(G^{\psi,p}_\theta(b)\right)
=b_\theta^{{\rm 1st}, p}\!\left((b^{\psi,p}_\theta)^{-1}(b)\right)
$
as in the proof of Lemma~\ref{lem:competition}.  Then, writing bids as functions of quantiles, it suffices to show, as in the proof of Theorem~\ref{thm:main}, that for each $\theta>\theta'$, $b_\theta^{{\rm 1st}, p}\!\left((b^{\psi,p}_\theta)^{-1}(b)\right)\geq b_{\theta'}^{{\rm 1st}, p}\!\left((b^{\psi,p}_{\theta'})^{-1}(b)\right)$ for any $b\in (b^{\psi,p}_{\theta}(0), b^{\psi,p}_{\theta'}(1))$. This follows from the same revenue-equivalence arguments that give the analog of (\ref{eq:rev-equivalence}) for every $x\in[ \underline x_\theta, 1)$, where $\underline x_\theta=F_\theta(c_\theta)$.
\end{proof}

As in the reserve-price setting in Section~\ref{sec:reserve}, the seller's revenue under entry fee $p$ at each $\theta$ can be written as 
\[
u(p,\theta)=\int_{c_\theta(p)}^{\overline v_\theta}\left(v-\frac{1-F_\theta(v)}{f_\theta(v)}\right)dF^n_\theta(v),
\] where the cutoff $c_\theta(p)$ solves $p=c_\theta(p)F_\theta(c_\theta(p))^{n-1}$.  
The choice of the entry fee is an instance of a monotone decision problem if the value distributions $(F_\theta)$ satisfy Myerson regularity and are monotone with respect to the star order, i.e., $\frac{F_{\theta}^{-1}(x)}{F_{\theta'}^{-1}(x)}$ is increasing in $x\in (0,1)$ for each $\theta>\theta'$.\footnote{Single-peakedness follows from Myerson regularity. Focusing on the interior solution for simplicity, the optimal fee $p^*_\theta$ at  $\theta$ is such that 
\[
p^*_\theta=c^*_\theta F_\theta(c^*_\theta)^{n-1}=F_\theta^{-1}(q_\theta^*)(q_\theta^*)^{n-1},
\] where $c^*_\theta$ is the unique type $v$ satisfying $v-\frac{1-F_\theta(v)}{f_\theta(v)}=0$ and $q_\theta^*:=F_\theta(c^*_\theta)$. 
Since $F_\theta^{-1}$ increases with $\theta$, it suffices to verify that $q^*_\theta$ is also increasing. 
To see this, note that $q_\theta^*$ also uniquely maximizes
$(1-q)F_\theta^{-1}(q)$. Indeed,
$
\frac{d}{dq}\bigl[(1-q)F_\theta^{-1}(q)\bigr]
=\frac{1-q}{f_\theta(F_\theta^{-1}(q))}-F_\theta^{-1}(q),
$
which changes sign from positive to negative at $q_\theta^*$ by Myerson regularity. By star-order monotonicity, for any $\theta>\theta'$,
\[
\frac{d}{dq}\log\bigl[(1-q)F_{\theta}^{-1}(q)\bigr]
\ge
\frac{d}{dq}\log\bigl[(1-q)F_{\theta'}^{-1}(q)\bigr].
\]
At $q=q_{\theta'}^*$ the right-hand side is zero. Hence, single-peakedness of $(1-q)F^{-1}_\theta(q)$
implies $q_{\theta}^*\ge q_{\theta'}^*$.
}  For example, the latter condition is satisfied by value distributions $(F_\theta)$ that belong to a scale family, i.e., $F_\theta(v)=F_0(v/s(\theta))$ for some $F_0$ and some increasing function $s(\cdot)>0$.

Based on this, one can apply Proposition~\ref{prop:entry} to an analogous two-period setting as in Section~\ref{sec:reserve}, where the auctioneer uses today's bid observations to optimally choose tomorrow's entry fee.

\section{Comparison against General Auctions without MLRP}\label{app:undominated}

By Theorem~\ref{thm:main}, the first-price bid distributions $G^{\rm 1st}$ are more accurate than the bid distributions $G^{\psi}$ under any symmetric monotone equilibrium of any standard auction $\psi$, provided that $G^{\psi}$ satisfies the MLRP. One issue in obtaining a comparison without the MLRP is that some standard auctions $\psi$ may admit multiple symmetric monotone equilibria, although equilibrium multiplicity does not arise under most canonical auction formats (e.g., $k$th-price and own-pay auctions; see Section~\ref{sec:subclass}). If $\psi$ admits multiple equilibria, one can in principle exploit arbitrary equilibrium selection across states to generate bid distributions $G^{\psi}$ that violate the MLRP and are more accurate than $G^{\rm 1st}$, broadly related to the way in which asymmetric equilibria under the second-price auction can perfectly reveal the state (see footnote~\ref{fn:report}). 

However, the following result shows that if we focus on standard auctions $\psi$ with a unique symmetric monotone equilibrium, then first-price auction bids are undominated in terms of the accuracy order. Thus, even if $G^{\psi}$ violates the MLRP, $G^{\psi}$ is never strictly more accurate than $G^{\rm 1st}$.


\begin{prop}\label{prop:undominated}
Suppose auction $\psi$ admits a unique symmetric monotone equilibrium, with bid distributions $G^{\psi}$. If $G^\psi$ is more accurate than $G^{{\rm 1st}}$, then
$G^{{\rm 1st}}$ is more accurate than $G^\psi$.
\end{prop}

To prove Proposition~\ref{prop:undominated}, we show that if $G^{\rm 1st}$ is not more accurate than $G^{\psi}$, then we can slightly perturb the equilibrium bidding function under $\psi$ to construct another equilibrium, contradicting uniqueness.

\begin{proof}
We write buyers' bids as functions of quantiles. 
Suppose that $G^\psi$ is more accurate than $G^{{\rm 1st}}$.  Then Lemma~\ref{lem:competition} implies that $ h_\theta(b):=b_\theta^\psi((b_\theta^{{\rm 1st}})^{-1}(b))$ is increasing in $\theta$ for every $b\geq 0$. It also implies that $b_\theta^\psi(0)$ and $b_\theta^\psi(1)$ are increasing in $\theta$. Each $h_\theta(\cdot)$ is continuous and strictly increasing in $b\in (b^{\rm 1st}_\theta(0), b^{\rm 1st}_\theta(1))$.

\medskip
\textbf{Step 1: Reverse hazard-rate order of $G^{\rm 1st}$.}
By the MLRP of $(F_\theta)$, it follows that $b^{\rm 1st}_\theta(0)$ and $b^{\rm 1st}_\theta(1)$ are increasing in $\theta$. For $b\in (b^{\rm 1st}_\theta(0), b^{\rm 1st}_\theta(1))$,  the first-order optimality condition in equilibrium can be rewritten as 
\[
 \frac{(G_\theta^{{\rm 1st}})'(b)}{G_\theta^{{\rm 1st}}(b)}
 =\frac{1}{(n-1)(F_\theta^{-1}((b_\theta^{{\rm 1st}})^{-1}(b))-b)}.
\]
The right-hand side is increasing in $\theta$, which shows that $G^{\rm 1st}_{\theta}$ is monotone with respect to the reverse hazard-rate order.

\medskip
\textbf{Step 2: Normalized expected payments under $\psi$.}
For each $\theta$, $b\in (b_\theta^{\rm 1st}(0), b_\theta^{\rm 1st}(1))$, and a function $h$ that is strictly increasing on $(b_\theta^{\rm 1st}(0), b_\theta^{\rm 1st}(1))$, we define
\[
 C_\theta[h](b):=\sum_i\frac{p_i((b_\theta^{{\rm 1st}})^{-1}(b))}{G^{\rm 1st}_\theta(b)^{n-1}}
 {\mathbb E}_\theta\!\left[
 \psi_i\bigl(h(b^{(1)}),\ldots, h(b^{(n)})\bigr) | b=b^{(i)} \right],
\]
where $b^{(1)},\ldots, b^{(n)}$ denote the order statistics of $n$ i.i.d.\ draws from distribution $G_\theta^{\rm 1st}$. 
Here, as in Appendix~\ref{app:uniform-order},  $p_i(x)$ denotes the probability that a realized quantile $x$ is $i$th-highest. 
When every buyer follows bidding function $h(b^{\rm 1st}_\theta(\cdot))$ in auction $\psi$,  $C_\theta[h](b)$ corresponds to the expected payment of quantile type $(b_\theta^{\rm 1st})^{-1}(b)$  divided by her winning probability $G^{\rm 1st}_\theta(b)^{n-1}$. 
Since $\psi$ is monotone in bids, $ C_\theta[h](b)$ is increasing in $h$ pointwise.

We claim that $C_\theta[h](b)$ is increasing in $\theta$. To prove this, define measures $\mu_{\theta,b}^-$ and $\mu_{\theta,b}^+$ on $\mathbb{R}_+$ by
\[
 d\mu_{\theta,b}^-(y)=\frac{\mathbf 1_{\{y<b\}}\,dG_\theta^{{\rm 1st}}(y)}{G_\theta^{{\rm 1st}}(b)},
 \qquad
 d\mu_{\theta,b}^+(y)=\frac{\mathbf 1_{\{y>b\}}\,dG_\theta^{{\rm 1st}}(y)}{G_\theta^{{\rm 1st}}(b)}.
\]
The first is a probability measure, which is FOSD-increasing in $\theta$ by Step 1. The second is a finite measure whose mass need not equal one. For each $z\geq b$,  
$
 \mu_{\theta,b}^+((z,\infty))
 =\frac{1-G_\theta^{{\rm 1st}}(z)}{G_\theta^{{\rm 1st}}(b)}
$
is increasing in $\theta$ by Step 1. This implies that integrals of nonnegative increasing functions with respect to $\mu_{\theta,b}^+$ are also increasing in $\theta$. 
Using the Markov property of order statistics,  we can write $C_\theta[h](b)$ as
\[
C_\theta[h](b)=\sum_{i=1}^{n}\binom{n-1}{i-1}
\int \int
 \psi_i\!\left(\bigl(h(b_j^+)\bigr)_{j=1}^{i-1},h(b),
\bigl(h(b_j^-)\bigr)_{j=1}^{n-i}
\right)
\prod_{j=1}^{n-i}d\mu_{\theta,b}^-(b_j^-)
\prod_{j=1}^{i-1}d\mu_{\theta,b}^+(b_j^+),
\]
which then shows that $C_\theta[h](b)$ is increasing in $\theta$.

For  $ h_\theta(\cdot)=b_\theta^\psi((b_\theta^{{\rm 1st}})^{-1}(\cdot))$, observe that $C_\theta[h_\theta](b)$ is the expected payment of quantile type $(b_\theta^{\rm 1st})^{-1}(b)$ in auction $\psi$ divided by the winning probability. Thus, the revenue-equivalence argument gives $ C_\theta[h_\theta](b)=b$ for $b\in (b^{\rm 1st}_\theta(0), b_\theta^{\rm 1st}(1))$.

\medskip
\textbf{Step 3: Construction of new strategy $ \widetilde b_{\theta'}$.}
Suppose toward a contradiction that $G^{\rm 1st}$ is not more accurate than $G^{\psi}$. By Lemma~\ref{lem:competition}, there exist $\theta>\theta'$ and a bid $b^*\ge0$ such that $b_\theta^{\rm 1st}((b_\theta^{\psi})^{-1}(b^*))<b_{\theta'}^{\rm 1st}((b_{\theta'}^{\psi})^{-1}(b^*))$.

Choose $b \in(b_\theta^{\rm 1st}((b_\theta^{\psi})^{-1}(b^*)), b_{\theta'}^{\rm 1st}((b_{\theta'}^{\psi})^{-1}(b^*)))$. Then $b\in (b_\theta^{\rm 1st}(0), b_{\theta'}^{\rm 1st}(1))$. Since the first-price auction bidding functions are strictly increasing, 
$
 h_{\theta'}(b)<b^*<h_\theta(b).
$
Choose a small compact interval $K\subseteq (b_\theta^{\rm 1st}(0), b_{\theta'}^{\rm 1st}(1))$ around $b$ such that 
$
 \sup_{b'\in K}h_{\theta'}(b')<b^*<\inf_{b'\in K}h_\theta(b').
$
Then there exists a function $\widetilde h$ that is continuous and strictly
increasing on $(b_{\theta'}^{\rm 1st}(0), b_{\theta'}^{\rm 1st}(1))$ such that
\[
 h_{\theta'}\le\widetilde h\le h_\theta,
 \qquad
 \widetilde h=h_{\theta'}\text{ outside }K,
 \qquad \text{ and } \quad
 \widetilde h\ne h_{\theta'}.
\]
Indeed, $\widetilde h$ can be constructed by taking a strictly positive and continuously differentiable
function $\eta$ supported in the interior of $h_{\theta'}(K)$ and setting 
$\widetilde h(b')=h_{\theta'}(b')+\varepsilon\eta(h_{\theta'}(b'))$.
For all sufficiently small $\varepsilon>0$, the map
$y\mapsto y+\varepsilon\eta(y)$ is strictly increasing and we have $\tilde h\leq h_\theta$ on $K$. Moreover, $\tilde h\leq h_\theta$ outside $K$ since $h_\theta(\cdot)\geq h_{\theta'}(\cdot)$. 
Define a new bidding strategy $
 \widetilde b_{\theta'}(x)
 :=\widetilde h(b_{\theta'}^{{\rm 1st}}(x))
$  under auction $\psi$ at $\theta'$.

\medskip
\textbf{Step 4: Verify that $ \widetilde b_{\theta'}$ is an equilibrium under $\psi$.} We verify that every type's expected payment under $ \widetilde b_{\theta'}$ is identical to the one under $b^\psi_{\theta'}$:
\begin{enumerate}
\item Take any quantile $x<1$ such that $b:=b^{\rm 1st}_{\theta'}(x)>b_\theta^{\rm 1st}(0)$. Step 2 and $h_{\theta'}\leq \tilde h\leq h_\theta$ give
\[
 b=C_{\theta'}[h_{\theta'}](b)
 \le C_{\theta'}[\widetilde h](b)
 \le C_{\theta'}[h_\theta](b)
 \le C_\theta[h_\theta](b)=b.
\]
This verifies the claim, since $C_{\theta'}[\widetilde h](b)$ and $C_{\theta'}[h_{\theta'}](b)$ are $x$'s expected payments normalized by the same winning probability under $ \widetilde b_{\theta'}$  and $b^\psi_{\theta'}$, respectively. 

\item Take any quantile $x$ such that $b_{\theta'}^{\rm 1st}(x)\le b_\theta^{\rm 1st}(0)$. Then $
\widetilde b_{\theta'}(x)=b_{\theta'}^\psi(x)\le b_\theta^\psi(0),
$ and all opponents' bids under either strategy are strictly below $b_\theta^\psi(1)$. Whenever $x$ does not win the item, her payment  must be zero under both strategies.  To see this, suppose that her payment is $\delta>0$ for some realization in which her rank is $i\geq 2$. Then choose $x'<1$ such that $b_\theta^\psi(x')$ exceeds every bid in that realization. Since $\psi_i$ is monotone, a state-$\theta$ bidder of quantile $y<x'$ pays at least $\delta$ whenever $i-1$ opponents have quantiles above $x'$ and the remaining $n-i$ have quantiles below $y$. 
The fact that each buyer's interim expected equilibrium payoff is bounded below by $0$ implies 
$
\delta\binom{n-1}{i-1}(1-x')^{i-1}y^{n-i}
\le \overline v_\theta y^{n-1},
$
which is impossible when $y\approx 0$. Whenever $x$ wins, all opponents' first-price bids are below $b_{\theta'}^{\rm 1st}(x)$, hence outside $K$, so the entire bid profile and payment are unchanged across  $ \widetilde b_{\theta'}$ and $b^\psi_{\theta'}$. 

\end{enumerate}

The strategy $\widetilde b_{\theta'}$ is continuous and strictly increasing.
Each type has no incentive to mimic other types' bids, as the resulting allocation probability and expected
payment coincide with those under $b_{\theta'}^\psi$. Likewise, bidding outside the range $\widetilde b_{\theta'}([0,1])$ is also not profitable. Hence, $\widetilde b_{\theta'}$ is another equilibrium under $\theta'$, which contradicts equilibrium uniqueness.
\end{proof}

\end{document}